\documentclass{article}

\usepackage{algorithm}
\usepackage{algorithmic}
\usepackage[a4paper]{geometry}
\usepackage{amssymb,amsmath,mathtools,xcolor,graphicx,xspace,colortbl,ragged2e,rotating} %
\usepackage{amsfonts} 
\usepackage{amsmath} 
\usepackage{amssymb} 
\usepackage{amsthm} 
\usepackage{arydshln}
\usepackage{color} 
\usepackage{graphicx} 
\usepackage{caption}
\usepackage{subcaption}
\usepackage{mathtools} 
\usepackage{hyperref} 
\usepackage{enumerate}
\usepackage{authblk}
\newcommand{\adots}{\mathinner{\mkern2mu \raisebox{0.1em}{.}\mkern2mu\raisebox{0.4em}{.}
\mkern2mu\raisebox{0.7em}{.}\mkern2mu}}
\DeclareGraphicsExtensions{.pdf,.eps,.ps,.png,.jpg,.jpeg}
\providecommand{\U}[1]{\protect\rule{.1in}{.1in}}
\allowdisplaybreaks
\newtheorem{theorem}{Theorem}

\newtheorem{assumption}[theorem]{Assumption}

\newtheorem{definition}[theorem]{Definition}
\newtheorem{example}[theorem]{Example}

\newtheorem{lemma}[theorem]{Lemma}
\newtheorem{notation}[theorem]{Notation}
\newtheorem{problem}[theorem]{Problem}
\newtheorem{proposition}[theorem]{Proposition}
\newtheorem{remark}[theorem]{Remark}

\usepackage{cite}
\usepackage{amsmath,amssymb}
\usepackage{tikz}
\usetikzlibrary{decorations.pathreplacing}

\usepackage{ytableau}

\begin{document}

\title{An  Algorithm for Discriminating the Complete Multiplicities of a Parametric Univariate Polynomial }
\author[a]{Simin Qin}
\author[b]{Bican Xia}
\author[a]{Jing Yang\thanks{%
Corresponding author: yangjing0930@gmail.com.}}
\affil[a]{{\small HCIC--School of Mathematical Sciences,
Center for Applied Mathematics of Guangxi,
Guangxi Minzu University, Nanning 530006, China}}
\affil[b]{\small School of Mathematical Sciences, Peking University, Beijing 100091, China}

\date{}
\maketitle

\begin{abstract}
In this paper, we tackle the parametric complete multiplicity problem for a univariate polynomial. Our approach to the parametric complete multiplicity problem has a significant difference from the classical method, which relies on repeated gcd computation. Instead, we introduce a novel technique that uses incremental gcds of the given polynomial and its high-order derivatives. This approach, formulated as non-nested subresultants, sidesteps the exponential expansion of polynomial degrees in the generated condition. We also uncover the hidden structure between the incremental gcds and pseudo-remainders. Our analysis reveals that the conditions produced by our new algorithm are simpler than those generated by the classical approach in most cases. Experiments show that our algorithm is faster than the one based on repeated gcd computation for problems with relatively big size.\end{abstract}

\section{Introduction}
Solving univariate polynomials for roots is a very fundamental problem in computer algebra with numerous applications.  Depending on whether the coefficients contain indeterminates, we categorize the problem into two sub-problems. 
\begin{itemize}
\item For polynomials with constant coefficients, ``solving" means finding/isolating the roots. 
\item For polynomials with parametric coefficients, ``solving" means classifying the \textit{complete root structures} the given polynomial may have and providing a set of conditions under which each structure occurs. By complete root structure, we mean the 
numbers of real and imaginary roots and their multiplicities. 
Thus complete root structure is also called \textit{complete multiplicity structure}.
\end{itemize}

In this paper, we tackle the second problem. More explicitly, we consider a polynomial of the form $P=a_nx^n+\cdots+a_0$, where 
$a_n\ne0$ and $a_i$'s are parameters which take values over the real field $\mathbb{R}$, and 
explore the problem of determining a set of conditions on $a_i$'s (which provides a partition of the parameter set) so that $P$ exhibits a specific complete multiplicity structure under each condition. 
For instance, we consider a quartic polynomial $P=a_{4}x^{4}+a_{3}x^{3}+a_{2}x^{2}+a_{1}x+a_{0}$, where 
$a_4\ne0$ and $a_i$'s are real parameters. All the potential complete multiplicity structures of $P$ are:
\[
\begin{array}{llll}
\left((1,1,1,1);
 ()\right), &\left((1,1); (1,1)\right), &\left((2,1,1); ()\right), &\left((2); (1,1)\right),\\ \left((2,2); ()\right), &\left((); (2,2)\right), &\left((3,1); ()\right),\ &\left((4); ()\right),
\end{array}\] 
where $(\boldsymbol{\mu_R};\boldsymbol{\mu_I})$ with $\# \boldsymbol{\mu_R}+\#\boldsymbol{\mu_I}=n$ is interpreted as follows:
\begin{itemize}
\item 
the real roots of $P$ have the multiplicity structure $\boldsymbol{\mu_R}$, and 
\item 
the imaginary roots of $P$ have the multiplicity structure $\boldsymbol{\mu_I}$. 
\end{itemize}
For instance, if $P$ has a complete multiplicity structure $((2);(1,1))$, then $P$ has a double real root and a pair of simple imaginary roots. 
With this setting, we aim to find a set of conditions $C_0,C_1,\ldots$ on  $a_{i}$'s so that
\[
\text{the complete multiplicity structure of $P$}=\left\{
\begin{array}{ll}
((1,1,1,1);()) & \text{iff $C_0$ holds} \\
((1,1);(1,1)) & \text{iff $C_1$ holds} \\
\vdots & \vdots \\
((4);()) & \text{iff $C_7$ holds}
\end{array}\right.
\]

\noindent In general, the problem is stated as follows:

\medskip
\noindent\textbf{Problem}: \emph{For every} $\boldsymbol{\mu_c}=\left(  \boldsymbol{\mu_R};\boldsymbol{\mu_I}\right)$ \textit{where} $\boldsymbol{\mu_R}=(\mu_{R,1},\ldots,\mu_{R,m_1})$ \textit{and} $\boldsymbol{\mu_I}=(\mu_{I,1},\ldots,\mu_{I,m_2})$ \textit{are such that} 
\begin{itemize}
\item $\mu_{R,1}\ge\ldots\ge\mu_{R,m_1}>0$,
\item $\mu_{I,1}=\mu_{I,2}\ge\ldots\ge\mu_{I,m_2-1}=\mu_{I,m_2}>0$, \emph{and}
\item $\sum_{i=1}^{m_1}\mu_{R,i}+\sum_{i=1}^{m_2}\mu_{I,i}=n,$

\end{itemize}
\emph{find a necessary and sufficient condition on the
coefficients of a polynomial }$P$ \emph{{over $\mathbb{R}$} of degree~}$n$ \textit{such
that the complete multiplicity structure of} $P$ \textit{is} $\boldsymbol{\mu_c}$.
\medskip

This problem is significant as it has numerous applications in various fields, such as mathematics, science, and engineering. Due to its importance, the
problem and several related problems have already been extensively studied
(e.g., see \cite{1998_Gonzalez_Recio_Lombardi,2021_Hong_Yang,2006_Liang_Jeffrey,2008_Liang_Jeffrey_Maza,1999_Liang_Zhang,1996_Yang_Hou_Zeng}).
In \cite{1996_Yang_Hou_Zeng}, Yang, Hou and Zeng gave an algorithm to generate a condition for discriminating different complete multiplicity structures of a univariate polynomial (referred to as YHZ's condition hereinafter) by making use of repeated gcd computation for parametric polynomials \cite{1971_Brown_Traub,1967_Collins,1983_Loos}. 
A similar idea was adopted by Gonzalez-Vega et al. \cite{1998_Gonzalez_Recio_Lombardi} for solving the real root classification and quantifier elimination problems by using Sturm-Habicht sequences. Roughly speaking, this method computes multiple factors at different levels, which can be realized by repeated gcd computation, i.e., computing the gcd of $P$ and its first derivative, the gcd of the previous gcd and its derivative, and so on. 
For the gcd at each level, we determine the condition that it has a certain number of distinct real roots/pairs of imaginary roots.
By conjoining all the conditions at different levels, one can get a condition that $P$ has a given complete multiplicity structure.
It should be pointed out that in YHZ's method, the polynomials in the conditions are computed from repeated gcds whose coefficients are nested determinants. Thus, the ``size" of these polynomials increases dramatically when the degree of $P$ grows, causing a huge computational burden.

A coarser version of the problem, which is stated as the parametric multiplicity problem, has been well studied (mainly via repeated gcd computation) in classical subresultant theory (e.g., \cite{1999_Yap}). In the parametric multiplicity structure, we do not differentiate real and imaginary roots, and thus, the multiplicity structure considered is called a complex multiplicity structure. Recently,  Hong and Yang revisited the problem in \cite{2021_Hong_Yang_JSC,2024_Hong_Yang:non-nested},  resulting in two non-nested conditions for determining the multiplicity structure. It is shown that the generated conditions by \cite{2021_Hong_Yang_JSC,2024_Hong_Yang:non-nested} are smaller than that produced via repeated gcd computation. Note that the complete multiplicity can be viewed as the refinement of complex multiplicity. We aim to find conditions with ``small" size for the complete multiplicity problem, measured in terms of the number of polynomials appearing in the conditions and their maximum degree.

The main contribution of this paper is to provide such a condition,
which has a
\emph{smaller} size than that in the previous methods. The key idea for generating the new condition is to replace the repeated gcds with the gcds of $P$ and its high-order derivatives of different orders for describing the multiple factors at different levels. For this purpose, we introduce the concept of incremental gcd and prove that it can be written as some specific subresultant of $P$ and its derivatives (see Theorem \ref{thm:icgcd}). Then, we devise an algorithm for computing the condition for every possible complete multiplicity structure $P$ may have. 
It is shown that the output condition has a smaller number of polynomials, and the maximal degree of polynomials in the condition is also significantly smaller than those in YHZ's method. Furthermore, we also explore the relationship among these incremental gcds and identify an interesting structure (see Proposition \ref{prop:sres_prem}) which is similar to the generalized Habicht's theorem in \cite{2024_Hong_Yang}.

The paper is structured as follows. 
In Section \ref{sec:problem}, we first present the problem to be addressed in a formal way. In Section \ref{sec:preliminaries}, we review the concept of subresultant for multiple polynomials and its equivalent form in roots. 
In Section \ref{sec:main}, we present the main result of the paper (Theorem \ref{thm:icgcd}), which is followed by a detailed proof in Section \ref{sec:proof}. 
In Section \ref{sec:algorithm}, we design an algorithm for solving the proposed problem with the newly developed tool. 
In Section \ref{sec:comparison}, we compare the performance of the algorithm and the  size of polynomials output by the algorithm and those
given by previous works. 
To keep the presentation of the main result in a tight manner,
we postpone the proofs for two secondary results (which we hope could be useful for tackling other related problems) to Appendices A and B.

\section{Problem Statement}\label{sec:problem}
\begin{notation}\label{notation_poly}
\
\begin{itemize}
\item $P = \sum_{i = 0}^{n}a_{i}x^{i}$, where $a_{n} \ne 0$;
\item $\operatorname{mult}(P) = (\mu_{1}, \dots, \mu_{m})$ is the multiplicity vector of $P$, where $\mu_{1} \ge \cdots \ge \mu_{m} \ge 1$. 
\end{itemize}
\end{notation}

Without loss of generality, we assume that the polynomial considered in this paper has degree greater than $1$.

\begin{definition} [Complete multiplicity]
 Given $P \in \mathbb{R}[x]$, assume $P$ has
 $m_1$ distinct real roots of $P$ with multiplicities $\mu_{R,1}, \dots, \mu_{R,m_{1}}$, and 
 $m_2$ distinct imaginary roots of $P$ with multiplicities $\mu_{I,1}, \dots, \mu_{I,m_{2}}$, respectively,
where $\mu_{R,1}\ge\dots\ge\mu_{R,m_{1}}\ge1$ and $\mu_{I,1}=\mu_{I,2}\ge\dots\ge\mu_{I,m_{2}-1}=\mu_{I,m_{2}}\ge1$. Then the \emph{complete multiplicity} of $P$, written as $\operatorname{cmult}(P)$, is defined by
 $$ \operatorname{cmult}(P) = ((\mu_{R,1}, \dots, \mu_{R,m_{1}}); (\mu_{I,1}, \dots, \mu_{I,m_{2}})). $$
\end{definition}

\begin{example}
 Let $P = x^5 - 4x^4 + 6x^3 - 6x^2 + 5x - 2$. Then $\operatorname{mult}(P) = (2, 1, 1, 1)$, since it can be verified that $P = (x - 1)^2(x - 2)(x^2 + 1)$. It is easy to see that
\[\operatorname{cmult}(P) =((2,1);(1,1)).\]
\end{example}

\begin{notation}\
\begin{itemize}
    \item $\mathcal{M}(n) :=\{(\mu_{1}, \dots, \mu_{m}) : \mu_{1} + \cdots + \mu_{m} = n, \mu_{1} \ge \cdots \ge \mu_{m} \ge 1\}$;

    \item $\overline{\mathcal{M}}(n):=$
\\
$\left\{
((\mu_{R,1}, \dots, \mu_{R,m_{1}}); (\mu_{I,1}, \dots, \mu_{I,m_{2}})):\bigwedge\left(
\begin{array}{l}
\mu_{R,1}\ge\dots\ge\mu_{R,m_{1}}\ge1\\
\mu_{I,1}=\mu_{I,2}\ge\dots\ge\mu_{I,m_{2}-1}=\mu_{I,m_{2}}\ge1\\
\sum_{i=1}^{m_1}\mu_{R,i}+\sum_{i=1}^{m_2}\mu_{I,i}=n
\end{array}
\right)
\right\}$.
\end{itemize}
\end{notation}

Obviously, every $\boldsymbol{\mu_c}\in\overline{\mathcal{M}}(n)$ can be viewed as a $2$-partition of some $\boldsymbol{\mu}\in\mathcal{M}(n)$ with the first part representing the multiplicities of real roots and the second part representing those of imaginary roots. 
For example, $\operatorname{cmult}(P) = ((2, 1); (1, 1))\in\overline{\mathcal{M}}(5)$ is a $2$-partition of $\operatorname{mult}(P) = (2, 1, 1, 1)\in\mathcal{M}(5)$, where $(2,1)$ indicates that $P$ has two real roots with one of them to be of multiplicity $2$ and $(1,1)$  indicates that $P$ has a pair of simple imaginary roots.

Now we are ready to give a formal statement of the problem addressed in the current paper.

\begin{problem}[Parametric  complete multiplicity problem]\label{problem}
\ 
\begin{description}
\item[In:\ \ ] $P=\sum_{i=0}^na_ix^i\in\mathbb{Z}[a_0,\ldots,a_n][x]$ where $n\ge2$ and $a_i$'s take values over $\mathbb{R}$ with $a_n$ assumed to be nonzero.
\item[Out:] for each $\boldsymbol{\mu_c}\in\overline{\mathcal{M}}(n)$, find a necessary and sufficient condition $C_{\boldsymbol{\mu_c}}$ on $a_i$'s such that $\operatorname*{cmult}(P)=\boldsymbol{\mu_c}$.
\end{description}
\end{problem}

\section{Preliminaries}\label{sec:preliminaries}

A useful tool for formulating the multiple factors of a polynomial is subresultant for multiple polynomials, which is proposed by Hong and Yang in \cite{2021_Hong_Yang, 2024_Hong_Yang} with several variants developed in \cite{2023_Wang_Yang} and is usually defined in the form of determinant polynomial.

\begin{definition} [Determinant polynomial]
 Let $\boldsymbol{M}$ be a $p \times q$ matrix where $p \leq q$. Then the determinant polynomial of $\boldsymbol{M}$, written as $\operatorname*{dp}\boldsymbol{M}$, is defined as
 $$ \operatorname*{dp}\boldsymbol{M} = \sum_{i=0}^{q-p}\boldsymbol{M}^{(i)}x^{i} $$
 where $\boldsymbol{M}^{(i)} = \left|\boldsymbol{M}_{1}, \ldots, \boldsymbol{M}_{p - 1}, \boldsymbol{M}_{q- i}\right|$ and $\boldsymbol{M}_{k}$ stands for the $k$-th column of $\boldsymbol{M}$.
\end{definition}

\begin{notation}\label{notation}\
\begin{itemize}
\item $\boldsymbol{F} = (F_0, F_1, \ldots, F_t)\subseteq \mathbb{Z}[a][x]$, where 
\begin{itemize}
    \item $t \geq 1$,
    \item $d_i=\deg F_i$, and
    
    \item $F_i = \sum_{j=0}^{d_i}a_{ij}x^j$;
\end{itemize}

\item $\mathcal{P}(d_0,t)=\{ (\delta_{1}, \ldots, \delta_{t})\in \mathbb{N}^{t}_{\geq 0}:\,  \delta_{1} + \cdots + \delta_{t}\le d_{0}\}$.

\item For $\boldsymbol{\delta}=(\delta_1,\ldots,\delta_{t})\in\mathcal{P}(d_0,t)$, $|\boldsymbol{\delta}|:=\delta_1+\cdots+\delta_{t}$.
\end{itemize}
\end{notation}

\begin{definition} [Generalized Sylvester matrix]
 Given $\boldsymbol{F}$  as in Notation \ref{notation} and $\boldsymbol{\delta}\in\mathcal{P}(d_0,t)$, the $\boldsymbol{\delta}$-th Sylvester matrix of $\boldsymbol{F}$, written as $\boldsymbol{M}_{\boldsymbol{\delta}}(\boldsymbol{F})$, is defined as
 $$ \boldsymbol{M}_{\boldsymbol{\delta}}(\boldsymbol{F}) =
\begin{bmatrix}
 \boldsymbol{B}_{0}\\
 \boldsymbol{B}_{1} \\
 \vdots \\
 \boldsymbol{B}_{t}
\end{bmatrix}$$
 where
\begin{align}
 \boldsymbol{B}_{i} =&
\begin{bmatrix}
 \cdots & \cdots & a_{i1} & a_{i0} & & \\
 & \ddots & & \ddots & \ddots & \\
 & & \cdots & \cdots & a_{i1} & a_{i0}
\end{bmatrix}_{\delta_i\times(\delta_0+d_0)}\notag\\
 \delta_0=&
 \left\{\begin{array}{cl}
 \max_{\substack{1\le i\le t\\\delta_i\ne0}}(d_i+\delta_i)-d_0,&\text{if}\ \ \max_{\substack{1\le i\le t\\\delta_i\ne0}}(d_i+\delta_i)\ge d_0,\\[10pt]
 1,&\text{otherwise}.
\end{array}\right.\label{eqs:delta0}
\end{align}
\end{definition}

\begin{remark}
 One may check that the number of rows in $\boldsymbol{M}_{\boldsymbol{\delta}}(\boldsymbol{F})$ is $\delta_0+|\boldsymbol{\delta}|$ and that of columns is $\delta_0+d_0$. Since $|\boldsymbol{\delta}|\le d_0$, we immediately see that $\boldsymbol{M}_{\boldsymbol{\delta}}(\boldsymbol{F})$ is either a wide or square matrix. Thus it is meaningful to compute its determinant polynomial. 
\end{remark}

\begin{definition}
The $\boldsymbol{\delta}$-th subresultant $R_{\boldsymbol{\delta}}$ of $\boldsymbol{F}$ with respect $x$ is defined as
$$ R_{\boldsymbol{\delta}}(\boldsymbol{F}) := \operatorname{dp}\boldsymbol{M}_{\boldsymbol{\delta}}(\boldsymbol{F}). $$
The principal leading coefficient of $R_{\boldsymbol{\delta}}(\boldsymbol{F})$, i.e., the coefficient of the term $x^{d_0-|\boldsymbol{\delta}|}$, is called the $\boldsymbol{\delta}$-th principal subresultant coefficient and is denoted by $\overline{R_{\boldsymbol{\delta}}(\boldsymbol{F})}$.
\end{definition}

\begin{remark}\
\begin{itemize}
    \item When the meaning of $\boldsymbol{F}$ is clear from the context, $R_{\boldsymbol{\delta}}(\boldsymbol{F})$ and $\overline{R_{\boldsymbol{\delta}}(\boldsymbol{F})}$ can be abbreviated as $R_{\boldsymbol{\delta}}$ and $\overline{R_{\boldsymbol{\delta}}}$, respectively.
    
    \item If the length of $\boldsymbol{\delta}$ is $1$, that is, $\boldsymbol{\delta}=(\delta_1)$, we also write $R_{\boldsymbol{\delta}}$ as $R_{\delta_1}$ for simplicity.
\end{itemize}
\end{remark}

\smallskip
In \cite{2021_Hong_Yang}, Hong and Yang provided an equivalent expression in the roots of $F_0$ for $R_{\boldsymbol{\delta}}(\boldsymbol{F})$. To present the formula, we introduce the following notations.

\begin{notation}\label{notation:root}\
\begin{itemize}
\item $\boldsymbol{\alpha}=(\alpha_1,\ldots,\alpha_{d_0})$ where $\alpha_{1}, \ldots, \alpha_{d_{0}}$ are the complex roots of $F_{0}$;

\item $\boldsymbol{X}_{\varepsilon}=(x^{\varepsilon},\ldots,x^0)^T$;

\item $\boldsymbol{X}_{\varepsilon}(\alpha_i)=(\alpha_i^{\varepsilon},\ldots,\alpha_i^0)^T$;

\item $\boldsymbol{X}_{\varepsilon}(\boldsymbol{\alpha})=\begin{bmatrix}
\boldsymbol{X}_{\varepsilon}(\alpha_1)&\cdots&\boldsymbol{X}_{\varepsilon}(\alpha_{d_0})\end{bmatrix}$;
\item $ V(\boldsymbol{\alpha}) = |\boldsymbol{X}_{d_0-1}(\boldsymbol{\alpha})|$.
\end{itemize}
\end{notation}

\begin{theorem}\label{thm:Subrespoly}
Given $P$ and $\boldsymbol{\delta}\in\mathcal{P}(d_0,t)$ as in Notation \ref{notation}, we have $$ R_{\boldsymbol{\delta}}(\boldsymbol{F}) =(-1)^{\varepsilon}\cdot a^{\delta_{0}}_{0d_{0}}\cdot
 \begin{vmatrix}\boldsymbol{M}_1\\\vdots\\\boldsymbol{M}_t\\\boldsymbol{X}_{\varepsilon}(\boldsymbol{\alpha})&\boldsymbol{X}_{\varepsilon}\end{vmatrix}\Big/ V(\boldsymbol{\alpha}) $$
 where
\begin{itemize}
\item $\delta_0$ is as in \eqref{eqs:delta0}, $\varepsilon = d_{0} - |\boldsymbol{\delta}| $,
\item $\boldsymbol{M}_i=\begin{bmatrix}
(x^{\delta_{i} - 1}F_{i})(\alpha_{1}) & \cdots & (x^{\delta_{i} - 1}F_{i})(\alpha_{d_{0}}) & \\
 \vdots & & \vdots & \\
 (x^{0}F_{i})(\alpha_{1}) & \cdots & (x^{0}F_{i})(\alpha_{d_{0}})
 \end{bmatrix}$, and
\item $\boldsymbol{X}_{\varepsilon}$, $\boldsymbol{X}_{\varepsilon}(\boldsymbol{\alpha})$ and $V(\boldsymbol{\alpha})$ are as in Notation \ref{notation:root}.
\end{itemize}
\end{theorem}

\section{Main Results}\label{sec:main}

\begin{definition} [Conjugate]
 Let $\boldsymbol{\mu}= (\mu_{1}, \dots, \mu_{m}) \in \mathcal{M}(n)$. Then the conjugate of $\boldsymbol{\mu}$ is defined by $\boldsymbol{\bar{\mu}} = (\bar{\mu}_{1}, \dots, \bar{\mu}_{n})$ where
$$\bar{\mu}_{i} = \# \{j \in [1, \dots, m] : \mu_{j} \geq i \}.$$
\end{definition}

\begin{definition} [Incremental gcd]
 Given $\boldsymbol{F}$ as in Notation \ref{notation}, let
 \[G_i = \gcd(F_0,F_1,\ldots,F_i).\]
 Then we call $\boldsymbol{G} = (G_1,\ldots,G_n)$ the \emph{incremental gcd} of $\boldsymbol{F}$, written as $\operatorname{icgcd} \boldsymbol{F}$.
\end{definition}

\begin{theorem}[Main theorem]\label{thm:icgcd}
 Given $P\in\mathbb{R}[x]$ of degree $n$ with $\operatorname*{mult}(P) = \boldsymbol{\mu}$, assume $\boldsymbol{\bar{\mu}} = (\bar{\mu}_{1},\ldots, \bar{\mu}_{n})$. Let $\boldsymbol{F} = (P^{(0)},P^{(1)},\ldots,$
$P^{(n)})$ where $P^{(k)}$ is the $k$-th derivative of $P$ with respect to $x$ for $0 \le k \le n$. Then we have
 \[\operatorname*{icgcd} \boldsymbol{F} = (G_1, \ldots, G_n)\]
 where 
 \[G_i = R_{(\bar{\mu}_{1},\ldots, \bar{\mu}_{i})}\big(P^{(0)},P^{(1)},\ldots,
P^{(i)}\big)\]
\end{theorem}

\begin{remark}
It is observed that when $\delta_i=0$, $F_i$ is not involved in $R_{\boldsymbol{\delta}}(\boldsymbol{F})$. Thus for simplicity, we make  the convention that $$R_{(\delta_{1},\ldots, \delta_{i})}(\boldsymbol{F}):=R_{(\delta_{1},\ldots, \delta_{i})}(F_0,\ldots,F_i)$$ 
In other words, we view $\delta_{i+1}=\dots=\delta_n=0$. Under this convention, $G_i$ in Theorem \ref{thm:icgcd} can be written as 
\[G_i=R_{(\bar{\mu}_{1},\ldots, \bar{\mu}_{i})}(\boldsymbol{F})\]
where $\boldsymbol{F}=\big(P^{(0)},P^{(1)},\ldots,$
$P^{(n)}\big)$. Moreover, from now on, we will assume $\boldsymbol{F}=\big(P^{(0)},P^{(1)},\ldots,$
$P^{(n)}\big)$ in the rest of the paper unless specified in the context. 
\end{remark}

\begin{example}\label{ex:icgcd}
Consider $P = (x - 1)^{3}(x +1)^{2} (x^{2}+1)= x^7 - x^6 - x^5 + x^4 - x^{3} + x^{2} + x - 1$. Let $\boldsymbol{F} = (P^{(0)}, P^{(1)}, P^{(2)}, P^{(3)}, P^{(4)}, P^{(5)}, P^{(6)}, P^{(7)})$. It is easy to see that 
\begin{equation}\label{eqs:ex_icgcd}
\operatorname{icgcd}\boldsymbol{F}=\left((x - 1)^2(x + 1), x-1, 1, 1, 1, 1, 1\right)
\end{equation}
In what follows, we calculate $\operatorname{icgcd}\boldsymbol{F}$ by Theorem \ref{thm:icgcd}.

Since $\boldsymbol{\mu}=(3,2,1,1)$, $\boldsymbol{\bar{\mu}} = (\bar{\mu}_{1}, \bar{\mu}_{2}, \bar{\mu}_{3}, \bar{\mu}_{4}, \bar{\mu}_{5}, \bar{\mu}_{6}, \bar{\mu}_{7}) = (4, 2, 1, 0, 0, 0, 0)$. Then
\[
\begin{array}{rl}
G_1 &= R_{(4)}(\boldsymbol{F}) \\[2pt]
&= \operatorname*{dp}(x^{2}P^{(0)}, x^{1}P^{(0)}, x^{0}P^{(0)}, x^{3}P^{(1)}, x^{2}P^{(1)}, x^{1}P^{(1)}, x^{0}P^{(1)}) \\[2pt] 
&= \operatorname*{dp}\begin{bmatrix}
1& -1 & -1 & 1 & -1 & 1 & 1 & -1 & & \\
 & 1& -1 & -1 & 1 & -1 & 1 & 1 & -1 &\\
 & & 1& -1 & -1 & 1 & -1 & 1 & 1 & -1 \\
7 & -6 & -5 & 4 & -3 & 2 & 1 & & & \\
 & 7 & -6 & -5 & 4 & -3 & 2 & 1 & & \\
 & & 7 & -6 & -5 & 4 & -3 & 2 & 1 & \\
 & & & 7 & -6 & -5 & 4 & -3 & 2 & 1
\end{bmatrix}\\
\\[-6pt]
&= -1536(x - 1)^2(x + 1) \\[2pt]
 G_2 &= R_{(4,2)}(\boldsymbol{F}) \\[2pt]
 &= \operatorname*{dp}(x^{2}P^{(0)}, x^{1}P^{(0)}, x^{0}P^{(0)}, x^{3}P^{(1)}, x^{2}P^{(1)}, x^{1}P^{(1)}, x^{0}P^{(1)}, x^{1}P^{(2)}, x^{0}P^{(2)})\\[2pt]
&= \operatorname*{dp}\begin{bmatrix}
1& -1 & -1 & 1 & -1 & 1 & 1 & -1 & & \\
& 1& -1 & -1 & 1 & -1 & 1 & 1 & -1 &\\
& & 1& -1 & -1 & 1 & -1 & 1 & 1 & -1 \\
7 & -6 & -5 & 4 & -3 & 2 & 1 & & & \\
& 7 & -6 & -5 & 4 & -3 & 2 & 1 & & \\
& & 7 & -6 & -5 & 4 & -3 & 2 & 1 & \\
& & & 7 & -6 & -5 & 4 & -3 & 2 & 1 \\
& & & 42 & -30 & -20 & 12 & -6 & 2 & \\
& & & & 42 & -30 & -20 & 12 & -6 & 2 
\end{bmatrix} \\\\[-6pt]
&= -1179648(x - 1)\\[2pt]
 G_3 &= R_{(4,2,1)}(\boldsymbol{F}) \\[2pt]
 &= \operatorname*{dp}(x^{2}P^{(0)}, x^{1}P^{(0)}, x^{0}P^{(0)}, x^{3}P^{(1)}, x^{2}P^{(1)}, x^{1}P^{(1)}, x^{0}P^{(1)}, x^{1}P^{(2)}, x^{0}P^{(2)}, x^{0}P^{(3)})\\[2pt]
&= \operatorname*{dp}\begin{bmatrix}
1& -1 & -1 & 1 & -1 & 1 & 1 & -1 & & \\
& 1& -1 & -1 & 1 & -1 & 1 & 1 & -1 &\\
& & 1& -1 & -1 & 1 & -1 & 1 & 1 & -1 \\
7 & -6 & -5 & 4 & -3 & 2 & 1 & & & \\
& 7 & -6 & -5 & 4 & -3 & 2 & 1 & & \\
& & 7 & -6 & -5 & 4 & -3 & 2 & 1 & \\
& & & 7 & -6 & -5 & 4 & -3 & 2 & 1 \\
& & & 42 & -30 & -20 & 12 & -6 & 2 & \\
& & & & 42 & -30 & -20 & 12 & -6 & 2 \\
& & & & & 210 & -120 & -60 & 24 & -6
\end{bmatrix} \\\\[-6pt]
&= -56623104 \\[2pt]
G_4 &= R_{(4,2,1,0)}(\boldsymbol{F})\ \ \ \ \ \  = R_{(4,2,1)}(\boldsymbol{F}) \\[2pt]
G_5 &= R_{(4,2,1,0,0)}(\boldsymbol{F})\ \ \ \ \, = R_{(4,2,1)}(\boldsymbol{F})\\[2pt]
G_6 &= R_{(4,2,1,0,0,0)}(\boldsymbol{F})\ \ \; = R_{(4,2,1)}(\boldsymbol{F})\\[2pt]
G_7 &= R_{(4,2,1,0,0,0,0)}(\boldsymbol{F})\  = R_{(4,2,1)}(\boldsymbol{F}) 
\end{array}
\]
Therefore,
\[\operatorname*{icgcd}\boldsymbol{F}=(G_1,G_2,G_3,G_4,G_5,G_6,G_7)\]
which only differs from \eqref{eqs:ex_icgcd} by constant factors.
\end{example}

\begin{remark}\ 
\begin{itemize}
\item It is noted that when $\bar{\mu}_i\ne0$ and $\bar{\mu}_{i+1}=\cdots=\bar{\mu}_{n}=0$, $$G_i=R_{(\bar{\mu}_{1},\ldots, \bar{\mu}_{i})}(\boldsymbol{F})$$
is a nonzero constant. Obviously, in this case,$$G_{i+1}=\cdots=G_{n}=G_i$$
\item Since $k!$ is a common factor of the coefficients for $P^{(k)}$, we can use $P^{(k)}/k!$ instead of $P^{(k)}$ to simplify the computation in practice. So the gcds in Example \ref{ex:icgcd} can be simplified into the followings:
\begin{align*}
G_1 &= -1536(x - 1)^2(x + 1),\quad G_2 = -589824(x - 1),\quad \\
G_3 &= G_4 = G_5 = G_6 = G_7 = -9437184 
\end{align*}
\end{itemize}
\end{remark}

One may wonder whether there are inherent relationships among $R_{\boldsymbol{\delta}}(\boldsymbol{F})$'s. Indeed, we succeed to find such a relationship for polynomials with formal coefficients, which reveals the hidden structures in the potential  $G_i$'s and converts the computation of subresultants into that of pseudo-remainders. We hope this relationship could be useful for exploring more hidden structures among the potential $G_i$'s and 
 can be applied to enhance the efficiency of computing  $R_{\boldsymbol{\delta}}(\boldsymbol{F})$'s.
The proof of the proposition can be found in Appendix A.

\begin{proposition}\label{prop:sres_prem}
Let $P=\sum_{i=0}^na_ix^i\in\mathbb{Z}[a_0,\ldots,a_n][x]$ and $\boldsymbol{F}=(P^{(0)},P^{(1)},\ldots,$
$P^{(n)})$. Let $\boldsymbol{\delta}=(\delta_1,\ldots,\delta_t)\in\mathcal{P}(n,t) $ be such that \begin{itemize}
\item $\delta_1\ge\cdots\ge\delta_t\ge1$, and
\item $\delta_j>\delta_{j+1}$ holds for some $1\le j<t$. 
\end{itemize}
Then 
\[\overline{R_{\boldsymbol{\delta}-\boldsymbol{{e}}_j-\boldsymbol{{e}}_t}(\boldsymbol{F})}\cdot R_{\boldsymbol{\delta}}(\boldsymbol{F})=\operatorname*{prem}(R_{\boldsymbol{\delta}-\boldsymbol{{e}}_j}(\boldsymbol{F}),R_{\boldsymbol{\delta}-\boldsymbol{{e}}_t}(\boldsymbol{F}))\]
where $\boldsymbol{{e}}_k$ is the $k$-th unit vector of length $t$.
\end{proposition}

\begin{example}
Consider $P=a_3x^3+a_2x^2+a_1x+a_0$ and $\boldsymbol{\delta}= (\delta_{1}, \delta_{2}) = (2, 1)$.
Suppose we have
\[
 R_{(1, 1)} 
(\boldsymbol{F})= 6\,a_{3}( 3\,a_{3}x+a_{2}) ,\ \ 
 R_{(2, 0)} (\boldsymbol{F})= a_{3} \big( (6\,a_{{1}}a_{{3}}-2\,{a_{{2}}^{2}})x+(9\,a_{{0}}a_{{3}}-
a_{{1}}a_{{2}}) \big) 
,\ \ \overline{R_{(1,0)}(\boldsymbol{F})}= 3\,a_{3}.
\]
By Proposition \ref{prop:sres_prem}, we can calculate $R_{(2, 1)}(\boldsymbol{F})$ via
\begin{align*}
R_{(2, 1)}
(\boldsymbol{F})&= \operatorname{prem}(R_{(1,1)}(\boldsymbol{F}), R_{(2, 0)}(\boldsymbol{F}))\big/\overline{R_{(1,0)}(\boldsymbol{F})}\\
&= -6\,a_{{3}}^{2} \left( 27\,a_{{0}}a_{{3}}^{2}-9\,a_{{1}}a_{{2}}a_{
{3}}+2\,a_{{2}}^{3} \right) 
/(3\,a_{3})\\
&=-2\,a_{{3}} \left( 27\,a_{{0}}a_{{3}}^{2}-9\,a_{{1}}a_{{2}}a_{{3}}+2
\,{a_{{2}}^{3}} \right) 
\end{align*}
\end{example}

\section{Proof of Theorem \ref{thm:icgcd}}
\label{sec:proof}

This section is devoted to proving the paper's main result (Theorem \ref{thm:icgcd}). To achieve this goal, we first rewrite $R_{\boldsymbol{\delta}}(\boldsymbol{F})$ 
from an expression in coefficients to that in multiple roots through divided difference. Then we simplify the resulting expression in roots with the multiplicity information (i.e., the evaluation of the $k$-th derivative of $P$ at its multiple root $r_i$ is zero for $0\le k\le \mu_i-1$ and is nonzero when $k=\mu_i$ where $\mu_i$ is the multiplicity of $r_i$). It turns out that the simplified expression corresponds to a certain gcd in $\operatorname{icgcd}\boldsymbol{F}$.

\subsection{Converting \texorpdfstring{$R_{\boldsymbol{\delta}}(\boldsymbol{F})$}{} into an expression in multiple roots}
\begin{lemma}[Subresultant in multiple roots] \label{lemma:Sdelta}
 Let $P$ be of degree $n$ with $m$ distinct roots $r_{1}, \dots, r_{m}$ whose multiplicities are $\mu_{1},\ldots,\mu_{m}$, respectively.
 Let $\boldsymbol{F} = (P^{(0)},P^{(1)},\ldots,P^{(n)})$ and $\boldsymbol{\delta} = (\delta_{1}, \ldots, \delta_{n})\in\mathcal{P}(n,n)$. Then we have
 $$ R_{\boldsymbol{\delta}}(\boldsymbol{F}) = c_{\boldsymbol{\delta},\boldsymbol{\mu}} \cdot \frac{\begin{vmatrix}
 \boldsymbol{H}^{(0)}(r_1)&\cdots&\boldsymbol{H}^{(\mu_1-1)}(r_1)&\cdots&\cdots&\boldsymbol{H}^{(0)}(r_m)&\cdots&\boldsymbol{H}^{(\mu_m-1)}(r_m)&\\
 \boldsymbol{X}_{\varepsilon}^{(0)}(r_1)&\cdots&\boldsymbol{X}_{\varepsilon}^{(\mu_1-1)}(r_1)&\cdots&\cdots&\boldsymbol{X}_{\varepsilon}^{(0)}(r_m)&\cdots&\boldsymbol{X}_{\varepsilon}^{(\mu_m-1)}(r_m)&\boldsymbol{X}_{\varepsilon}
 \end{vmatrix}}{\prod_{1 \leq u < v \leq m}^{} (r_{v} - r_{u})^{\mu_{u}\mu_{v}}}$$
 where
\begin{itemize}
 \item $ c_{\boldsymbol{\delta},\boldsymbol{\mu}} = (- 1)^{\binom{n}{2} + \varepsilon}\dfrac{a^{\delta_{0}}_{n}}{\prod_{u = 1}^{m} \prod_{v = 0}^{\mu_{u} - 1} v!}$,
 
\item $\varepsilon=n-|\boldsymbol{\delta}|$,

 \item $\delta_0$ is as in \eqref{eqs:delta0},
 \item $\boldsymbol{X}_{\varepsilon}=(x^{\varepsilon},\ldots,x^0)^T$,  \item $\boldsymbol{X}_{\varepsilon}^{(k)}=\left((x^{\varepsilon})^{(k)},\ldots,(x^0)^{(k)}\right)^T$, and
 \item $\boldsymbol{H}^{(k)}=\left(\left(x^{\delta_{1} - 1}P^{(1)}\right)^{(k)},\ldots,\left(x^{0}P^{(1)}\right)^{(k)},\ldots\ldots,\left(x^{\delta_{n} - 1}P^{(n)}\right)^{(k)},\ldots,\left(x^{0}P^{(n)}\right)^{(k)}\right)^T$.
\end{itemize}
\end{lemma}

\begin{proof}
 Let $P = a_{n}(x - \alpha_{1}) \cdots (x - \alpha_{n})$.   When $\alpha_{1}, \dots, \alpha_{n}$ are treated as numbers, without loss of generality, we assume that $\alpha_{1}, \dots, \alpha_{n}$ are grouped into $m$ sets as follows:
 \begin{align*}
 I_{1} :=&\ \{\alpha_{1} \ \dots \ \dots \ \dots \ \dots \ \alpha_{\mu_{1}}\},\\
 I_{2} :=&\ \{\alpha_{\mu_{1} + 1} \ \dots \ \dots \ \dots \ \alpha_{\mu_{1} + \mu_{2}}\},\\
 \vdots& \\
 I_{m} :=&\ \{\alpha_{\mu_{1} + \dots + \mu_{m - 1} + 1} \ \dots \ \alpha_{\mu_{1} + \dots + \mu_{m - 1} + \mu_{m}}\},
 \end{align*}
 where elements in $I_{i}$ are all equal to $r_{i}$.
 
Now we treat $\alpha_{1}, \dots, \alpha_{n}$ as indeterminates.  Let
$$\boldsymbol{H}=\left(x^{\delta_{1} - 1}P^{(1)},\ldots,x^{0}P^{(1)},\ldots,x^{\delta_{n} - 1}P^{(n)},\ldots,x^{0}P^{(n)}\right)^T$$ and $\boldsymbol{X}_{\varepsilon}=(x^{\varepsilon},\ldots,x^0)^T$. By Theorem \ref{thm:Subrespoly}, $$ R_{\boldsymbol{\delta}}(\boldsymbol{F})=(-1)^{\varepsilon}\cdot a^{\delta_{0}}_{n} \cdot
 \begin{vmatrix}
 \boldsymbol{H}(\alpha_1)&\cdots&\boldsymbol{H}(\alpha_{n})&\\
 \boldsymbol{X}_{\varepsilon}(\alpha_1)&\cdots&\boldsymbol{X}_{\varepsilon}(\alpha_n)& \boldsymbol{X}_{\varepsilon}
 \end{vmatrix}\Big/ V(\boldsymbol{\alpha}) $$
 where $\delta_0$ is as in \eqref{eqs:delta0}, $\varepsilon=n - | \boldsymbol{\delta}|$, and
 $ \boldsymbol{X}_{\varepsilon}$ and $V(\boldsymbol{\alpha})$ are as in Notation \ref{notation:root}.
 It follows that
\[
R_{\boldsymbol{\delta}}\ =\ \frac{(-1)^{\varepsilon}\cdot a^{\delta_{0}}_{n} \cdot
\begin{vmatrix}
 \boldsymbol{H}(\alpha_1)&\cdots&\boldsymbol{H}(\alpha_{n})&\\
 \boldsymbol{X}_{\varepsilon}(\alpha_1)&\cdots&\boldsymbol{X}_{\varepsilon}(\alpha_n)&\boldsymbol{X}_{\varepsilon}
 \end{vmatrix}}{V(\boldsymbol{\alpha})}
\ =\ \frac{(-1)^{\varepsilon}\cdot a^{\delta_{0}}_{n} \cdot
\begin{vmatrix}
 \boldsymbol{H}(\alpha_1)&\cdots&\boldsymbol{H}(\alpha_{n})&\\
 \boldsymbol{X}_{\varepsilon}(\alpha_1)&\cdots&\boldsymbol{X}_{\varepsilon}(\alpha_n)&\boldsymbol{X}_{\varepsilon}
 \end{vmatrix}}{(- 1)^{\binom{n}{2}}\prod_{1 \leq u < v \leq n}(\alpha_v-\alpha_u)}
\]

  Next, we carry out the exact division so that the differences between the collapsed $\alpha_i$'s do not appear in the denominator.
  For this purpose, we let $Q[x_{1}, \dots, x_{u}]$ denote the ($u - 1$)-th divided difference of $Q \in \mathbb{R}[x]$ at $x_{1}, x_{2}, \dots, x_{u}$ which is defined recursively as follows:
 $$ Q[x_{1}, \dots, x_{u}] = \left\{\begin{array}{ll}
 Q(x_{1}), & \text{if} \ \ u = 1; \\
 \dfrac{Q[x_{1}, \dots, x_{u - 2}, x_{u}] - Q[x_{1}, \dots, x_{u - 2}, x_{u - 1}]}{x_{u} - x_{u - 1}}, & \text{if} \ \ u > 1.
 \end{array}\right.$$
Now we eliminate the factors of the form $\alpha_v-\alpha_u$ where $\alpha_u,\alpha_v\in I_1$ by using successive divided differences. For the sake of simplicity, let $c'=(-1)^{\varepsilon+{\binom{n}{2}}}\cdot a^{\delta_{0}}_{n} $. Then we have
\begin{align*}
R_{\boldsymbol{\delta}}\ =\ &\frac{c'
\begin{vmatrix}
 \boldsymbol{H}(\alpha_1)&\cdots&\boldsymbol{H}(\alpha_{n})&\\
 \boldsymbol{X}_{\varepsilon}(\alpha_1)&\cdots&\boldsymbol{X}_{\varepsilon}(\alpha_n)&\boldsymbol{X}_{\varepsilon}
 \end{vmatrix}}{\prod_{1 \leq u < v \leq n}(\alpha_v-\alpha_u)}\\
 =\ &\frac{c'
\left|\begin{array}{llll|lll|l}
 \boldsymbol{H}[\alpha_{1}]&\boldsymbol{H}[\alpha_{1}, \alpha_{2}]&\cdots&\boldsymbol{H}[\alpha_{\mu_{1} - 1}, \alpha_{\mu_{1}}]&\boldsymbol{H}[\alpha_{\mu_{1} + 1}]&\cdots&\boldsymbol{H}[\alpha_{n}]&\\
\boldsymbol{X}_{\varepsilon}[\alpha_{1}]&\boldsymbol{X}_{\varepsilon}[\alpha_{1}, \alpha_{2}]&\cdots&\boldsymbol{X}_{\varepsilon}[\alpha_{\mu_{1} - 1}, \alpha_{\mu_{1}}]&\boldsymbol{X}_{\varepsilon}[\alpha_{\mu_{1} + 1}]&\cdots&\boldsymbol{X}_{\varepsilon}[\alpha_{n}]&\boldsymbol{X}_{\varepsilon}
 \end{array}\right|}
 {\prod_{\underset{v - u > 1}{\alpha_{u}, \alpha_{v} \in I_{1}}}^{}(\alpha_{v} - \alpha_{u})\prod_{\underset{v - u > 0}{\alpha_{u}, \alpha_{v} \notin I_{1}}}^{}(\alpha_{v} - \alpha_{u})\prod_{\underset{\alpha_{v} \notin I_{1}}{\alpha_{u} \in I_{1}}}^{}(\alpha_{v} - \alpha_{u})}\\
\ =\ &\frac{c'
\left|
\setlength{\arraycolsep}{2pt}
\begin{array}{lllll|lll|l}
 \boldsymbol{H}[\alpha_{1}]&\boldsymbol{H}[\alpha_{1}, \alpha_{2}]&\boldsymbol{H}[\alpha_{1}, \alpha_{2},\alpha_{3}]&\cdots&\boldsymbol{H}[\alpha_{\mu_{1} - 2},\alpha_{\mu_{1} - 1}, \alpha_{\mu_{1}}]&\boldsymbol{H}[\alpha_{\mu_{1} + 1}]&\cdots&\boldsymbol{H}[\alpha_{n}]&\\
\boldsymbol{X}_{\varepsilon}[\alpha_{1}]&\boldsymbol{X}_{\varepsilon}[\alpha_{1}, \alpha_{2}]&\boldsymbol{X}_{\varepsilon}[\alpha_{1}, \alpha_{2},\alpha_{3}]&\cdots&\boldsymbol{X}_{\varepsilon}[\alpha_{\mu_{1} - 2},\alpha_{\mu_{1} - 1}, \alpha_{\mu_{1}}]&\boldsymbol{X}_{\varepsilon}[\alpha_{\mu_{1} + 1}]&\cdots&\boldsymbol{X}_{\varepsilon}[\alpha_{n}]&\boldsymbol{X}_{\varepsilon}
 \end{array}\right|}
 {\prod_{\underset{v - u > 2}{\alpha_{u}, \alpha_{v} \in I_{1}}}^{}(\alpha_{v} - \alpha_{u})\prod_{\underset{v - u > 0}{\alpha_{u}, \alpha_{v} \notin I_{1}}}^{}(\alpha_{v} - \alpha_{u})\prod_{\underset{\alpha_{v} \notin I_{1}}{\alpha_{u} \in I_{1}}}^{}(\alpha_{v} - \alpha_{u})}\\
\ \vdots\ &\\
=\ &\frac{c'
\left|\begin{array}{llll|lll|l}
 \boldsymbol{H}[\alpha_{1}]&\boldsymbol{H}[\alpha_{1}, \alpha_{2}]&\cdots&\boldsymbol{H}[\alpha_{1},\ldots, \alpha_{\mu_{1}}]&\boldsymbol{H}[\alpha_{\mu_{1} + 1}]&\cdots&\boldsymbol{H}[\alpha_{n}]&\\
\boldsymbol{X}_{\varepsilon}[\alpha_{1}]&\boldsymbol{X}_{\varepsilon}[\alpha_{1}, \alpha_{2}]&\cdots&\boldsymbol{X}_{\varepsilon}[\alpha_1,\ldots, \alpha_{\mu_{1}}]&\boldsymbol{X}_{\varepsilon}[\alpha_{\mu_{1} + 1}]&\cdots&\boldsymbol{X}_{\varepsilon}[\alpha_{n}]&\boldsymbol{X}_{\varepsilon}
 \end{array}\right|}
 {\prod_{\underset{v - u > 0}{\alpha_{u}, \alpha_{v} \notin I_{1}}}^{}(\alpha_{v} - \alpha_{u})\prod_{\underset{\alpha_{v} \notin I_{1}}{\alpha_{u} \in I_{1}}}^{}(\alpha_{v} - \alpha_{u})}
\end{align*}
 Repeating the same procedure for $\alpha_{v}'s$ in each $I_{u}$, for $u = 2, \dots, m$, successively, we have
$$R_{\boldsymbol{\delta}} =\dfrac{c'\left|
\setlength{\arraycolsep}{2pt}
\begin{array}{lll|l|lll|l}
\boldsymbol{H}[\alpha_{1}]&\cdots&\boldsymbol{H}[\alpha_{1}, \dots, \alpha_{\mu_{1}}]&\cdots&\boldsymbol{H}[\alpha_{\mu_{1} + \dots + \mu_{m - 1} + 1}]&\cdots &\boldsymbol{H}[\alpha_{\mu_{1} + \cdots + \mu_{m - 1} + 1}, \dots, \alpha_{n}]&\\
 \boldsymbol{X}_{\varepsilon}[\alpha_{1}]&\cdots&\boldsymbol{X}_{\varepsilon}[\alpha_{1}, \dots, \alpha_{\mu_{1}}]&\cdots&\boldsymbol{X}_{\varepsilon}[\alpha_{\mu_{1} + \dots + \mu_{m - 1} + 1}]&\cdots &\boldsymbol{X}_{\varepsilon}[\alpha_{\mu_{1} + \cdots + \mu_{m - 1} + 1}, \ldots, \alpha_{n}]&
\boldsymbol{X}_{\varepsilon}
\end{array}\right|}
 {\prod_{1 \leq u < v \leq m}^{}\prod_{\underset{\alpha_{q} \in I_{v}}{\alpha_{p} \in I_{u}}}^{}(\alpha_{q} - \alpha_{p})}
$$

Now we substitute
$$
\begin{array}{l}
r_1=\alpha_{1} = \dots = \alpha_{\mu_{1}}, \\
\ \ \ \ \vdots\\
r_m=\alpha_{\mu_{1} + \dots + \mu_{m - 1} + 1} = \dots = \alpha_{n}
\end{array}$$
 into $R_{\boldsymbol{\delta}}$ and obtain
 \begin{align}
 R_{\boldsymbol{\delta}} =\dfrac{c'
\left|\begin{array}{lll|l|lll|l}
 \boldsymbol{H}[r_{1}]&\cdots&\boldsymbol{H}[r_{1}, \dots, r_{1}]&\cdots&\boldsymbol{H}[r_{m}]&\cdots &\boldsymbol{H}[r_{m}, \dots, r_{m}]&\\
 \boldsymbol{X}_{\varepsilon}[r_{1}]&\cdots&\boldsymbol{X}_{\varepsilon}[r_{1}, \ldots, r_1]&\cdots&\boldsymbol{X}_{\varepsilon}[r_{m}]&\cdots &\boldsymbol{X}_{\varepsilon}[r_{m}, \dots, r_{m}]&
\boldsymbol{X}_{\varepsilon}
\end{array}\right|}
 {\prod_{1 \leq u < v \leq m}^{}(r_{v} - r_{u})^{\mu_{u}\mu_{v}}} \label{eqs: S1}
 \end{align}

 For any given polynomial $Q \in \mathbb{R}[x]$, we have
 $$ Q[\underset{k \ r_{u}'s}{\underbrace{r_{u}, \dots, r_{u}}}] = \frac{Q^{(k - 1)}(r_{u})}{(k - 1)!} $$

 Hence
\begin{align}
\boldsymbol{H}[\underset{k \ r_{u}\text{'s}}{\underbrace{r_{u}, \dots, r_{u}}}] \ =\
\dfrac{1}{(k - 1)!}&\left((x^{\delta_{1} - 1}P^{(1)})^{(k - 1)}(r_{u}), \dots, (x^{0}P^{(1)})^{(k - 1)}(r_{u}),\right. \notag\\[-10pt]
&\ \ \ \ \ \ \ \ \ \ \ \cdots \cdots, \notag\\
&\left.\ \ \ (x^{\delta_{n} - 1}P^{(n)})^{(k - 1)}(r_{u}),\dots, (x^{0}P^{(n)})^{(k - 1)}(r_{u})\right)^{T}\notag\\
=\ \dfrac{1}{(k - 1)!}&\boldsymbol{H}^{(k-1)}(r_u)\label{eqs: H}\\
\boldsymbol{X}_{\varepsilon}[\underset{k \ r_{u}\text{'s}}{\underbrace{r_{u}, \dots, r_{u}}}] \ =\ \frac{1}{(k - 1)!}&\left((x^{\varepsilon})^{(k - 1)}(r_{u}), \dots, (x^{0})^{(k - 1)}(r_{u})\right)^T\notag\\
=\ \dfrac{1}{(k - 1)!}&\boldsymbol{X}_{\varepsilon}^{(k-1)}(r_u)\label{eqs: X3}
\end{align}

The substitution of \eqref{eqs: H} and \eqref{eqs: X3} into \eqref{eqs: S1} yields
 $$ R_{\boldsymbol{\delta}} = \dfrac{c_{\boldsymbol{\delta},\boldsymbol{\mu}}
\left|\begin{array}{lll|l|lll|l}
 \boldsymbol{H}^{(0)}(r_{1})&\cdots&\boldsymbol{H}^{(\mu_1-1)}(r_{1})&\cdots&\boldsymbol{H}^{(0)}(r_{m})&\cdots &\boldsymbol{H}^{(\mu_m-1)}(r_{m})&\\
 \boldsymbol{X}_{\varepsilon}^{(0)}(r_{1})&\cdots&\boldsymbol{X}_{\varepsilon}^{(\mu_1-1)}(r_{1})&\cdots&\boldsymbol{X}_{\varepsilon}^{(0)}(r_{m})&\cdots &\boldsymbol{X}_{\varepsilon}^{(\mu_m-1)}(r_{m})&
\boldsymbol{X}_{\varepsilon}
\end{array}\right|}
 {\prod_{1 \leq u < v \leq m}^{}(r_{v} - r_{u})^{\mu_{u}\mu_{v}}} $$
where $$c_{\boldsymbol{\delta},\boldsymbol{\mu}} =\frac{c'}{\prod_{u = 1}^{m}\prod_{v = 0}^{\mu_{u} - 1}v!}= (- 1)^{{\binom{n}{2}} + \varepsilon} \frac{a^{\delta_{0}}_{n}}{\prod_{u = 1}^{m}\prod_{v = 0}^{\mu_{u} - 1}v!}$$
The proof is completed. 
\end{proof}

\subsection{Property of confluent Vandermonde matrices}
It is seen in Lemma \ref{lemma:Sdelta} that  the numer of the expression $R_{\boldsymbol{\delta}}(\boldsymbol{F})$ in multiple roots is a generalization of the confluent Vandermonde matrix, which inspires us to consider this particular type of matrices.

Let $p$ and $q$ be positive integers. The $ p \times q$ Vandermonde block in terms of $x$ is defined as
$$ \boldsymbol{U}(x;p,q) = (c_{ij})_{\substack{1 \leq i \leq p\\1 \leq j \leq q}}$$
where $c_{ij} =\binom{i-1}{j-1}x^{i-j}$ with the convention $\binom{i-1}{j-1}:=0$ for $i<j$. 

\begin{definition}[Confluent Vandermonde matrix]
Given $\boldsymbol{x}=(x_1,\ldots,x_{\ell})$ and  $\boldsymbol{\tau}=(\tau_1,\ldots,\tau_{\ell})$,
the $\boldsymbol{\tau}$-\emph{confluent Vandermonde matrix} in terms of $\boldsymbol{x}$ is defined as the ${k}  \times {k} $ matrix
$$ \boldsymbol{V}(\boldsymbol{x}; \boldsymbol{\tau}) := \begin{bmatrix}
\boldsymbol{U}(x_{1}; {k} , \tau_{1})& \cdots &\boldsymbol{U}(x_{\ell}; {k} ,\tau_{\ell})
\end{bmatrix} $$
where ${k} =\tau_1+\cdots+\tau_{\ell}$ and $\boldsymbol{U}(x_i; {k} , \tau_{i})$ is the $ {k}  \times \tau_i$ Vandermonde block in terms of $x_i$.
\end{definition}

\begin{example}
Consider $\boldsymbol{x} = (x_1,x_2,x_3), \boldsymbol{\tau} = (3,1,2)$ and $k = |\boldsymbol{\tau}| = 6$. It is easy to see that 
\begin{align*}
\boldsymbol{V}(\boldsymbol{x};\boldsymbol{\tau}) &= \boldsymbol{V}((x_1,x_2,x_3);(3,1,2)) \\
&=\begin{bmatrix}
\boldsymbol{U}(x_1;6,3)& \boldsymbol{U}(x_2;6,1) &\boldsymbol{U}(x_3;6,2)
\end{bmatrix} \\
&=\begin{bmatrix}
1& 0 & 0 & 1 & 1 &0 \\
x_1& 1 & 0 & x_2 & x_3 &1 \\
x_1^2& 2x_1 & 1 & x_2^2 & x_3^2 &2x_3 \\
x_1^3& 3x_1^2 & 3x_1 & x_2^3 & x_3^3 &3x_3^2 \\
x_1^4& 4x_1^3 & 6x_1^2 & x_2^4 &x_3^4  &4x_3^3 \\
x_1^5& 5x_1^4 & 10x_1^3 & x_2^5 &x_3^5  &5x_3^4
\end{bmatrix}
\end{align*}
\end{example}

\begin{lemma}[\!\!\cite{1991_Horn_Johnson}] \label{eqs:V}
With the above settings, we have
 \[
{V}(\boldsymbol{x}; \boldsymbol{\tau}):=|\boldsymbol{V}(\boldsymbol{x}; \boldsymbol{\tau})| = \prod_{1 \leq i < j \leq \ell}(x_{j} - x_{i})^{\tau_{i}\tau_{j}}.
\]
\end{lemma}
\begin{remark}
In particular, when $\tau_{1} = \cdots = \tau_{\ell}=1$, $V(\boldsymbol{x}; \boldsymbol{\tau})= \prod_{1 \leq i < j \leq {k} }^{}(x_{j} - x_{i})$, which is the Vandermonde determinant and thus can be abbreviated as $V(\boldsymbol{x})$.
\end{remark}

\subsection{Proof of Theorem \ref{thm:icgcd}}
\begin{proof}[Proof of Theorem \ref{thm:icgcd}]
Let $\boldsymbol{\mu} = (\mu_1,\ldots,\mu_m)$, i.e., $P = a_{n}(x - r_{1})^{\mu_{1}}\dots(x - r_{m})^{\mu_{m}}$. It is easy to see that
$$G_{i} =\gcd(P^{(0)},\ldots,P^{(i)})
= \prod_{\mu_{j} > i}^{}(x - r_{j})^{\mu_{j} - i}.$$
On the other hand, let $\boldsymbol{\bar{\mu}} = (\bar{\mu}_{1}, \ldots, \bar{\mu}_{n}) $ where $\bar{\mu}_j=\#\{j:\,\mu_j\ge i\}$ and $\bar{\boldsymbol{\mu}}_i=(\bar{\mu}_1,\ldots,\bar{\mu}_i)$. We only need to show that $ R_{\bar{\boldsymbol{\mu}}_i}(P^{(0)}, \ldots,P^{(i)})=c \prod_{\mu_{j} > i}(x - r_{j})^{\mu_{j} - i}$ for some constant $c$. For the sake of simplicity, we will omit $(P^{(0)}, \ldots,P^{(i)})$ in the proof when $R_{\bar{\boldsymbol{\mu}}_i}$ is mentioned.

By Lemma \ref{lemma:Sdelta},
 \begin{align}
 R_{\bar{\boldsymbol{\mu}}_i}\,
=\,&R_{(\bar{\mu}_1,\ldots,\bar{\mu}_i,0,\ldots,0)}(\boldsymbol{F}) \notag\\
 =\,&
 \dfrac{c_{\bar{\boldsymbol{\mu}}_i}
\left|\begin{array}{lll|l|lll|l}
 \boldsymbol{H}^{(0)}(r_{1})&\cdots&\boldsymbol{H}^{(\mu_1-1)}(r_{1})&\cdots&\boldsymbol{H}^{(0)}(r_{m})&\cdots&\boldsymbol{H}^{(\mu_{m} - 1)}(r_{m})&\\
 \boldsymbol{X}_{\varepsilon}^{(0)}(r_{1})&\cdots&\boldsymbol{X}_{\varepsilon}^{(\mu_{1} - 1)}(r_{1})&\cdots&\boldsymbol{X}_{\varepsilon}^{(0)}(r_{m})&\cdots&\boldsymbol{X}_{\varepsilon}^{(\mu_{m} - 1)}(r_{m})&\boldsymbol{X}_{\varepsilon}
\end{array}\right|}
 {\prod_{1 \leq u < v \leq m}^{}(r_{v} - r_{u})^{\mu_{u}\mu_{v}}}
 \label{eqs:SSS}
 \end{align}
 for some constant $c_{\bar{\boldsymbol{\mu}}_i}$ where
 \begin{align*}
 \boldsymbol{H}^{(k)} &=\left(\begin{array}{ccc|c|ccc}
 \left(x^{\bar{\mu}_{1} - 1}P^{(1)}\right)^{(k)}&\cdots&\left(x^{0}P^{(1)}\right)^{(k)}&\cdots\cdots&\left(x^{\bar{\mu}_{i} - 1}P^{(i)}\right)^{(k)}&\cdots&\left(x^{0}P^{(i)}\right)^{(k)}
\end{array}\right)^T,\\
 \boldsymbol{X}_{\varepsilon} &=\left(x^{\varepsilon},\ldots,x^0\right)^T,\quad \boldsymbol{X}_{\varepsilon}^{(k)}=\left((x^{\varepsilon})^{(k)},\ldots,(x^0)^{(k)}\right)^T,\quad\text{and}\quad
 \varepsilon= n - |\boldsymbol{\bar{\mu}}_{i}|.
\end{align*}
After reordering the columns, we obtain
\begin{align}
R_{(\bar{\mu}_1,\ldots,\bar{\mu}_i)}\ &=\ \pm\dfrac{c_{\bar{\boldsymbol{\mu}}_i}
 \left|\begin{array}{llllll l|l}
 \boldsymbol{M}^{(1)}&\cdots& \boldsymbol{M}^{(i)}&\boldsymbol{M}^{(i+1)}
 \end{array}\right|}{\prod_{1 \leq u < v \leq m}^{}(r_{v} - r_{u})^{\mu_{u}\mu_{v}}}
\label{eqs:reordered_S}
\end{align}
where 
\[\boldsymbol{M}^{(j)}=\begin{bmatrix}
\boldsymbol{H}^{(\mu_{1} - j)}(r_{1})&\cdots&\boldsymbol{H}^{(\mu_{\bar{\mu}_j} - j)}(r_{\bar{\mu}_j})\\
\boldsymbol{X}_{\varepsilon}^{(\mu_{1} - j)}(r_{1})&\cdots&\boldsymbol{X}_{\varepsilon}^{(\mu_{\bar{\mu}_j} - j)}(r_{\bar{\mu}_j})
\end{bmatrix}
\]
for $j=1,\ldots,i$, and 
\[\boldsymbol{M}^{(i+1)}=\left[\begin{array}{lll|l|lll|l}
\boldsymbol{H}^{(0)}(r_{1})&\cdots&\boldsymbol{H}^{(\mu_{1} -( i+1))}(r_{1})&\cdots&\boldsymbol{H}^{(0)}(r_{\bar{\mu}_{i+1}})&\cdots&\boldsymbol{H}^{(\mu_{\bar{\mu}_{i+1}} - (i+1))}(r_{\bar{\mu}_{i+1}})&\\
\boldsymbol{X}_{\varepsilon}^{(0)}(r_{1})&\cdots&\boldsymbol{X}_{\varepsilon}^{(\mu_{1} - (i+1))}(r_{1})&\cdots&\boldsymbol{X}_{\varepsilon}^{(0)}(r_{\bar{\mu}_{i+1}})&\cdots&\boldsymbol{X}_{\varepsilon}^{(\mu_{\bar{\mu}_{i+1}} - (i+1))}(r_{\bar{\mu}_{i+1}})&\boldsymbol{X}_{\varepsilon}
\end{array}\right] 
\]
The column reordering is inspired by the following observation.

Since $P$ and its first $\mu_{v} - 1$ derivatives are equal to zero at $x = r_{v}$, by the Leibniz rule, we immediately know that for $u = 1, \dots, i$ and $v$ satisfying $\mu_{v} \geq u,$
 $$ (x^{\xi}P^{(u)})^{(\ell)}(r_v) = \left\{\begin{matrix}
 0& \text{if}\ \ell+u < \mu_{v};\\
 r^{\xi}_{v}P^{(\mu_{v})}(r_v)& \text{if}\ \ell+u = \mu_{v};\\
 *& \text{if}\ \ell+u > \mu_{v}.
 \end{matrix}\right.$$
Thus for $k = 0, \dots, \mu_{v} - 1$,
\begin{equation}\label{eqs:H_k(r_j)}
\boldsymbol{H}^{(k)}(r_v)=\left( \begin{array}{c|c|c|ccc|ccccc}\boldsymbol{0}_{1\times \bar{\mu}_1}&\cdots&\boldsymbol{0}_{1\times \bar{\mu}_{\mu_v-k-1}}
&r^{\bar{\mu}_{\mu_v-k}-1}_{v}P^{(\mu_v)}(r_{j})
& \cdots&r^{0}_{v}P^{(\mu_v)}(r_{v})
& *&\cdots&*
\end{array}\right)^T,
\end{equation}
Moreover, we can also derive that
\begin{align}
\boldsymbol{X}_{\varepsilon}^{(k)}(r_v)&=\left(\dfrac{\varepsilon!}{(\varepsilon-k)!}\cdot r^{\varepsilon - k}_{v}, \dots, {k!}\cdot r^{0}_{v}, 0, \dots, 0\right)^T=k!\left(\binom{\varepsilon}{k}\cdot r^{\varepsilon - k}_{v}, \dots,  \binom{k}{k}r^{0}_{v}, 0, \dots, 0\right)^T\label{eqs:X_k(r_j)} 
\end{align}

Next we further simplify $\boldsymbol{M}^{(j)}$'s by making use of \eqref{eqs:H_k(r_j)} and \eqref{eqs:X_k(r_j)} and obtain the followings.
\begin{enumerate}[(a)]
\item When $j\le i$,
we have
\begin{align*}
\boldsymbol{H}^{(\mu_1-j)}(r_1)\ &=\ 
\setlength{\arraycolsep}{2pt}
\left( \begin{array}{c|c|c|ccc|ccccc}\boldsymbol{0}_{1\times \bar{\mu}_1}&\cdots&\boldsymbol{0}_{1\times \bar{\mu}_{j-1}}
&r^{\bar{\mu}_{j}-1}_{1}P^{(\mu_1)}(r_{1})
& \cdots&r^{0}_{1}P^{(\mu_1)}(r_{1})
& *&\cdots&*
\end{array}\right)^T\\
&\ \,\vdots\\
\boldsymbol{H}^{(\mu_{\bar{\mu}_{j}}-j)}(r_{\bar{\mu}_j})\ &=\ \left( 
\setlength{\arraycolsep}{2pt}
\begin{array}{c|c|c|ccc|ccccc}\boldsymbol{0}_{1\times \bar{\mu}_1}&\cdots&\boldsymbol{0}_{1\times \bar{\mu}_{j-1}}
&r^{\bar{\mu}_{j}-1}_{\bar{\mu}_j}P^{(\mu_{\bar{\mu}_j})}(r_{\bar{\mu}_j})
& \cdots&r^{0}_{\bar{\mu}_j}P^{(\mu_{\bar{\mu}_j})}(r_{\bar{\mu}_j})
& *&\cdots&*
\end{array}\right)^T
\end{align*}
Thus
\[
\boldsymbol{M}^{(j)}=\left[
\begin{array}{c}
\boldsymbol{0}_{\bar{\mu}_1\times\bar{\mu}_j}\\
\vdots\\
\boldsymbol{0}_{\bar{\mu}_{j-1}\times\bar{\mu}_j}\\[3pt]
\boldsymbol{M}_j\\
*_{\bar{\mu}_{j+1}\times\bar{\mu}_j}\\
\vdots\\
*_{\bar{\mu}_{i}\times\bar{\mu}_j}
\end{array}\right]
\]
where
\begin{equation}\label{eqs:M_j}
\boldsymbol{M}_j=\begin{bmatrix}
r^{\bar{\mu}_{j}-1}_{1}P^{(\mu_1)}(r_{1})
& \cdots&r^{\bar{\mu}_{j}-1}_{\bar{\mu}_j}P^{(\mu_{\bar{\mu}_j})}(r_{\bar{\mu}_j})\\
\vdots
& &\vdots\\
r^{0}_{1}P^{(\mu_1)}(r_{1})
& \cdots&r^{0}_{\bar{\mu}_j}P^{(\mu_{\bar{\mu}_j})}(r_{\bar{\mu}_j})
\end{bmatrix}
\end{equation}

\item The simplification of $\boldsymbol{M}^{(i+1)}$ relies on the following observations:
\begin{itemize}
\item $\boldsymbol{H}^{(0)}(r_{\ell})=\cdots=\boldsymbol{H}^{(\mu_{\ell} - (i+1))}(r_\ell)=\boldsymbol{0}_{\sum_{j=1}^i{\bar{\mu}_j}\times1}=\boldsymbol{0}_{|{\boldsymbol{\bar\mu}}_i|\times1}$
for $\ell=1,\ldots,\bar{\mu}_{i+1}$, and
\item $\begin{bmatrix}
\boldsymbol{X}_{\varepsilon}^{(0)}(r_{\ell})&\cdots& \boldsymbol{X}_{\varepsilon}^{(\mu_{\ell} - (i+1))}(r_{\ell})
\end{bmatrix}=\boldsymbol{J}\boldsymbol{U}(r_{\ell};\varepsilon+1,\mu_{\ell}-i)\boldsymbol{D}_{\mu_{\ell}-i}$
\end{itemize}
where 
\[
\boldsymbol{J}=\begin{bmatrix}
    &&1\\
    &\adots&\\
    1&&
\end{bmatrix}_{(\varepsilon+1)\times(\varepsilon+1)}\qquad
\boldsymbol{D}_{\mu_{\ell}-i}=\begin{bmatrix}
    0!\\
    &\ddots&\\
    &&(\mu_{\ell}-(i+1))!
\end{bmatrix}_{(\mu_{\ell}-i)\times(\mu_{\ell}-i)}
\]
Hence
\begin{align*}
\boldsymbol{M}^{(i+1)}\ =&\ \left[\begin{array}{l|l|l|l}
\boldsymbol{0}_{|\boldsymbol{\bar{\mu}_i}|\times (\mu_1-i)}&\cdots&\boldsymbol{0}_{|\boldsymbol{\bar{\mu}_i}|\times (\mu_{\bar{\mu}_{i+1}}-i)}&\\
\boldsymbol{JU}(r_{1};\varepsilon+1,\mu_{1}-i)\boldsymbol{D}_{\mu_{1}-i}&\cdots&\boldsymbol{JU}(r_{\bar{\mu}_{i+1}};\varepsilon+1,\mu_{\bar{\mu}_{i+1}}-i)\boldsymbol{D}_{\mu_{\bar{\mu}_{i+1}}-i}&\boldsymbol{X}_{\varepsilon}
\end{array}\right]
\end{align*}
\end{enumerate}
Let $$\boldsymbol{N}=\begin{bmatrix}\boldsymbol{JU}(r_{1};\varepsilon+1,\mu_{1}-i)\boldsymbol{D}_{\mu_{1}-i}&\cdots&\boldsymbol{JU}(r_{\bar{\mu}_{i+1}};\varepsilon+1,\mu_{\bar{\mu}_{i+1}}-i)\boldsymbol{D}_{1}&\boldsymbol{X}_{\varepsilon}\end{bmatrix}$$
Using (a) and (b), we simplify \eqref{eqs:reordered_S} and get
\[
R_{(\bar{\mu}_1,\ldots,\bar{\mu}_i)}\ =\ \pm\dfrac{c_{\bar{\boldsymbol{\mu}}_i}
 \left|\begin{array}{llllll l|l}
 \boldsymbol{M}_1\\
 \vdots&\ddots\\
 \cdot&\cdots& \boldsymbol{M}_i\\
 \cdot&\cdots&\cdot&\boldsymbol{N}
 \end{array}\right|}{\prod_{1 \leq u < v \leq m}^{}(r_{v} - r_{u})^{\mu_{u}\mu_{v}}}
 =\pm\dfrac{c_{\bar{\boldsymbol{\mu}}_i}\prod_{j=1}^i|\boldsymbol{M}_j|\cdot |\boldsymbol{N}|
 }{\prod_{1 \leq u < v \leq m}^{}(r_{v} - r_{u})^{\mu_{u}\mu_{v}}}
\]
where $\boldsymbol{M}_j$ is as in \eqref{eqs:M_j} for $j=1,\ldots,i$.

It is easy to derive that
\[|\boldsymbol{M}_j|=\begin{vmatrix}
r^{\bar{\mu}_{j}-1}_{1}P^{(\mu_1)}(r_{1})
& \cdots&r^{\bar{\mu}_{j}-1}_{\bar{\mu}_j}P^{(\mu_{\bar{\mu}_j})}(r_{\bar{\mu}_j})\\
\vdots
& &\vdots\\
r^{0}_{1}P^{(\mu_1)}(r_{1})
& \cdots&r^{0}_{\bar{\mu}_j}P^{(\mu_{\bar{\mu}_j})}(r_{\bar{\mu}_j})
 \end{vmatrix}=\left(\prod_{1\le \ell\le \bar{\mu}_j}P^{(\mu_\ell)}(r_{\ell})\right)\cdot V(r_1,\ldots,r_{\bar{\mu}_j})\ne0\]
Furthermore, noting that
\begin{align*}
\boldsymbol{N}=\,&\boldsymbol{J}\begin{bmatrix}\boldsymbol{U}(r_{1};\varepsilon+1,\mu_{1}-i)&\cdots&\boldsymbol{U}(r_{\bar{\mu}_{i+1}};\varepsilon+1,\mu_{\bar{\mu}_{i+1}}-i)&\boldsymbol{JX}_{\varepsilon}\end{bmatrix}\\
&\cdot\operatorname{diag}\begin{bmatrix}
\boldsymbol{D}_{\mu_{1}-i}&\cdots&\boldsymbol{D}_{\mu_{\bar{\mu}_{i+1}}-i}&\boldsymbol{I}_1
\end{bmatrix}
\end{align*}
and $\boldsymbol{JX}_{\varepsilon}=\begin{bmatrix}
x^0&\cdots&x^{\varepsilon}
\end{bmatrix}^T$
where $\boldsymbol{I}_1$ is the identity matrix of size $1$,
we have
\begin{align*}
|\boldsymbol{N}|&\ =\ \pm \left(\prod_{j=1}^{\bar{\mu}_{i+1}}|\boldsymbol{D}_{\mu_j-i}|\right)V((r_1,\ldots,r_{\bar{\mu}_{i+1}},x); (\mu_{1}-i,\ldots,\mu_{\bar{\mu}_{i+1}}-i,1))\\
&\ =\ \pm \left(\prod_{j=1}^{\bar{\mu}_{i+1}}|\boldsymbol{D}_{\mu_j-i}|\right)\left(\prod_{1 \leq u < v \leq \bar{\mu}_{i+1}}(r_{v} - r_{u})^{(\mu_u-i)(\mu_v-i)}\right)\cdot \left(\prod_{1\le u\le \bar{\mu}_{i+1}}(x-r_u)^{\mu_u-i}\right).
\end{align*}
Therefore,
\begin{align*}
R_{(\bar{\mu}_1,\ldots,\bar{\mu}_i)} &= c\prod_{1\le u\le \bar{\mu}_{i+1}}(x-r_u)^{\mu_u-i}
=c\prod_{\mu_u>i}(x-r_u)^{\mu_u-i}=G_i\end{align*}
where 
$$c=\pm\dfrac{c_{\bar{\boldsymbol{\mu}}_i} \left(\prod_{j=1}^{\bar{\mu}_{i+1}}|\boldsymbol{D}_{\mu_j-i}|\right)\prod_{j=1}^i\left(V(r_1,\ldots,r_{\bar{\mu}_j})\prod_{1\le \ell\le \bar{\mu}_j}P^{(\mu_\ell)}(r_{\ell}) \right)\cdot \left(\prod_{1 \leq u < v \leq \bar{\mu}_{i+1}}(r_{v} - r_{u})^{(\mu_u-i)(\mu_v-i)}\right)}{\prod_{1 \leq u < v \leq m}^{}(r_{v} - r_{u})^{\mu_{u}\mu_{v}}} $$
which is obviously non-zero due to the assumption $r_1, \ldots, r_m$ are distinct.
\end{proof}

\section{Algorithm}\label{sec:algorithm}

In order to design an algorithm to compute the conditions for discriminating all the potential complete multiplicities $P$ may have, we need two more tools developed by Hong et al., which are non-nested multiplicity discriminant \cite{2024_Hong_Yang:non-nested} and discriminant sequence \cite{1996_Yang_Hou_Zeng}.\subsection{Non-nested multiplicity discriminants for univariate polynomials}

\begin{definition}[$\boldsymbol{\lambda}$-discriminant]
Let $P = \sum_{i = 0}^{n}a_{i}x^{i}\in\mathbb{R}[x]$ where $a_{n} \ne 0$ and $\boldsymbol{F} = (P^{(0)},$ $P^{(1)},\ldots,
P^{(n)})$ and $\boldsymbol{\lambda}\in\mathcal{M}(n)$. Then we call $R_{\boldsymbol{\lambda}}(\boldsymbol{F})$ the \emph{$\boldsymbol{\lambda}$-discriminant} of $P$.
\end{definition}

\begin{theorem}[\!\!\cite{2024_Hong_Yang:non-nested}]\label{thm:nonnested}
We have
$$\operatorname{mult}(P) = \boldsymbol{\mu}\ \ \ \Longleftrightarrow\ \ \ \boldsymbol{\bar{\mu}}=\max_{\substack{\boldsymbol{\lambda}\in\mathcal{M}(n)\\R_{\boldsymbol{\lambda}}(\boldsymbol{F})\ne0}}\boldsymbol{\lambda}, $$
where $\max$ is w.r.t. the lexicographical ordering.
\end{theorem}

\subsection{The discriminant sequence of a univariate polynomial}
In \cite{1996_Yang_Hou_Zeng}, Yang et al. proposed a method for determining the number of distinct real/imaginary roots for univariate polynomials with parametric coefficients.
In this subsection, we give a brief review on the work.

\begin{definition} [Discrimination matrix] Given a polynomial
$$ P(x) = a_{n}x^{n} + a_{n - 1}x^{n - 1} + \cdots + a_{1}x + a_{0},$$
we write the derivative of $P(x)$ as
$$ P'(x) = 0\cdot x^{n} +na_{n}x^{n-1} + (n-1)a_{n - 1}x^{n-2} + \cdots + a_{1}. $$
Then the Sylvester matrix of $P(x)$ and $P'(x)$,
$$ {\rm Syl}(P)=\begin{bmatrix}
 a_{n} & a_{n - 1} & a_{n - 2} & \cdots & a_{0} & & & \\
 & na_{n} & (n - 1)a_{n - 1} & \cdots & a_{1} & & & \\
 & a_{n} & a_{n - 1} & \cdots & a_{2} & a_{0} & & \\
 & & na_{n} & \cdots & 2a_{2} & a_{1} & & \\
 & & & \ddots & &&\ddots & \\
 & & & & a_{n} & a_{n - 1} & \cdots & a_{0}\\
 & & & & & na_{n} & \cdots &a_{1}
\end{bmatrix} $$
is called the discrimination matrix of $P(x)$.
\end{definition}

\begin{definition} [Discriminant sequence]
Let $D_{i}$ denote the $2i\times 2i$ principal minor of ${\rm Syl}(P)$ for $i = 1, \dots, n$.
Then the $n$-tuple
 $$ \left[ D_{1}, D_{2}, \dots, D_{n} \right ] $$
 is called the discriminant sequence of $P(x)$.
\end{definition}

In order to determine the number of distinct real/imaginary roots of $P$ with its discriminant sequence, we recall the concept of revised sign list \cite{1996_Yang_Hou_Zeng}. 

Given a sign list $[\sigma_{1}, \sigma_{2}, \dots, \sigma_{n}]\in\{-1,0,1\}^n$, its revised sign list
$ [\sigma'_{1}, \sigma'_{2}, \dots, \sigma'_{n}]$ can be constructed as follows.
\begin{itemize}
 \item If $[\sigma_{i}, \sigma_{i+1}, \dots, \sigma_{i+j}]$ is a section of the given list, where $\sigma_{i}\neq 0$, $\sigma_{i+1} = \cdots = \sigma_{i+ j -1} = 0$, $\sigma_{i+j}\neq 0$, then we replace the section
 $$ [\sigma_{i+1}, \sigma_{i+2}, \dots, \sigma_{i+j-1}] $$
 by
 $$ [-\sigma_{i}, -\sigma_{i}, \sigma_{i}, \sigma_{i}, -\sigma_{i}, -\sigma_{i}, \sigma_{i}, \sigma_{i}, \dots]; $$
 i.e. let
 $$ \sigma'_{i+k} = (-1)^{\frac{k+1}{2}}\cdot \sigma_{i}$$
 for $k= 1, 2, \dots, j-1$.
 \item Otherwise, let $\sigma'_{i} = \sigma_{i}$.
 \end{itemize}

\begin{theorem} \label{thm:YHZ}
 Given a polynomial $P\in\mathbb{R}[x]$,
 let $\sigma$ be the sign list of the discriminant sequence of $P$ and $\nu$ be the number of the sign changes of the revised sign list of $\sigma$. Then
 we have
 \begin{itemize}
\item the number of the pairs of distinct conjugate imaginary roots of $P$ is $\nu$;

\item
 the number of the distinct real roots of $P$ equals $\eta-2\nu$ where $\eta$ is the number of non-vanishing members in the revised sign list of $\sigma$.
\end{itemize}
\end{theorem}

\subsection{Algorithm}

Combining Theorems \ref{thm:icgcd}, \ref{thm:nonnested} and \ref{thm:YHZ}, we can derive a new necessary and sufficient condition for a polynomial $P=\sum_{i=0}^na_ix^i$ having a designated complete multiplicity $(\boldsymbol{\mu_R};\boldsymbol{\mu_I})$. This new condition is based on an observation that if $r_{k}$ is a root of $P$ with multiplicity $\mu_k$, then it is a root of $G_i$ with multiplicity $\max(\mu_k-i,0)$. 

Let $\boldsymbol{\mu}$ be the complex multiplicity of $P$ corresponding to $(\boldsymbol{\mu_R};\boldsymbol{\mu_I})$. By Theorem \ref{thm:YHZ}, with the premise $\operatorname{mult}(P)=\boldsymbol{\mu}$,  $G_i$ has $\nu$ pairs of distinct imaginary roots if and only if \[\nu={\rm Var}\left(
D_1(G_i,G_{i}'),\ldots,
D_{\bar{\mu}_{i+1}}(G_i,G_{i}')
\right)\]
where $G_0:=P$ and $\operatorname{Var}(L)$ is the short-hand for the number of sign changes for the revised list of the sign list of $L$. Therefore,
we have
\begin{align}
\operatorname*{cmult}\left(  P\right)  =(\boldsymbol{\mu_R};\boldsymbol{\mu_I})\ \ \ \Longleftrightarrow
\ \ \ &\left(\bigwedge_{\substack{\boldsymbol{\gamma}\succ\boldsymbol{\bar{\mu}}\\\boldsymbol{\gamma}\in\mathcal{M}(n)}}R_{\boldsymbol{\gamma}}(\boldsymbol{F})=0\right)  \ \ \ \wedge
\ \ R_{\boldsymbol{\bar{\mu}}}(\boldsymbol{F})\ne0\notag\\
 \wedge\ &\left(  \bigwedge\limits_{i=0}^{\mu_{1}-1}\ \ {\rm Var}\left(
 D_1(G_i,G_{i}'),\ldots,
 D_{\bar{\mu}_{i+1}}(G_i,G_{i}')
 \right)=\bar{\mu}_{I,i+1}/2\right)\label{eqs:QXYcond}
\end{align}
where $\bar{\mu}_{I,i}$ is the $i$-th element of $\boldsymbol{\bar{\mu}_I}$.

Now we propose
the following algorithm as a solution to  Problem \ref{problem}.

\begin{algorithm}[H]
\caption{${\sf ParametricCompleteMultiplicity}(P)$}\label{alg:rpm}
\begin{description}
\item[In:\ \ ] $P=\sum_{i=0}^na_ix^i\in\mathbb{R}[a_0,\ldots,a_n][x]$ 
\item[Out:] $\mathcal{C}=\left\{(\boldsymbol{\mu_{c}},  C_{\boldsymbol{\mu_{c}}}):\boldsymbol{\mu_{c}}\in\overline{\mathcal{M}}(n)\right\}$ such that $\operatorname{cmult}(P)=\boldsymbol{\mu_{c}}\ \Longleftrightarrow\ C_{\boldsymbol{\mu_{c}}}$.
\end{description}

\begin{algorithmic}[1]
\STATE $\boldsymbol{F}\leftarrow(P^{(0)},\ldots,P^{(n)})$\\ $\mathcal{C}\leftarrow\emptyset$

\STATE $\mathcal{M}(n)\leftarrow\left\{(\delta_1,\ldots,\delta_t):\,\delta_1+\cdots+\delta_t=n, \delta_1\ge\ldots\ge\delta_t\ge1\right\}$\\
 $\mathcal{N}(n)\leftarrow\left\{(\delta_1,\ldots,\delta_t):\,\delta_1+\cdots+\delta_t< n, \delta_1\ge\ldots\ge\delta_t\ge1\right\}$

\FOR{$\boldsymbol{\delta}\in\mathcal{M}(n)\cup\mathcal{N}(n)$}
\STATE $R_{\boldsymbol{\delta}}\leftarrow R_{\boldsymbol{\delta}}(\boldsymbol{F})$
\ENDFOR

\FOR{$\boldsymbol{\mu}\in\mathcal{M}(n)$}
\STATE $\boldsymbol{\bar{\mu}}=(\bar{\mu}_1,\bar{\mu}_2,\ldots)\leftarrow$ the conjugate of $\boldsymbol{\mu}$
\FOR{$i=0,\ldots,\mu_1-1$}
\STATE
$G_i\leftarrow \left\{
\begin{array}{ll}
P&\text{if}\ i=0\\
R_{(\bar{\mu}_1,\ldots,\bar{\mu}_i)}&\text{if}\ i>0
\end{array}\right.$
\FOR{$j=1,\ldots,\bar{\mu}_{i+1}$}
\STATE $
 D_{j}\leftarrow \text{the $j$-th polynomial in the discriminant sequence of}\ G_i$
\ENDFOR
\ENDFOR
\FOR{each 2-partition of $\boldsymbol{\mu}$, say $\boldsymbol{\mu_c}=(\boldsymbol{\mu_{R}};\boldsymbol{\mu_{I}})$,}
\STATE  $\boldsymbol{\bar{\mu}_I}=(\bar{\mu}_{I,1},\ldots)\leftarrow$ the conjugate of $\boldsymbol{\mu_I}$;

\STATE $\mathcal{C}(\boldsymbol{\mu_c})\ \ \ \leftarrow\ \ \ \left(\bigwedge_{\substack{\boldsymbol{\gamma}\succ\boldsymbol{\bar{\mu}}\\\boldsymbol{\gamma}\in\mathcal{M}(n)}}R_{\boldsymbol{\gamma}}(\boldsymbol{F})=0\right)  \ \ \ \wedge
\ \ \left(R_{\boldsymbol{\bar{\mu}}}(\boldsymbol{F})\ne0\right)$

\hspace{5em}$\wedge\ \left(  \bigwedge\limits_{i=0}^{\mu_{1}-1}\ \ {\rm Var}\left(
 D_1,\ldots,
 D_{\bar{\mu}_{i+1}}
 \right)=\bar{\mu}_{I,i+1}/2\right)$

\STATE $\mathcal{C}\leftarrow \mathcal{C}\bigcup\{(\boldsymbol{\mu_c},\mathcal{C}(\boldsymbol{\mu_c}))\}$
\ENDFOR
\ENDFOR
\RETURN $\mathcal{C}$.
\end{algorithmic}
\end{algorithm}

\begin{remark}
\rm\text{In} \rm\cite{1996_Yang_Hou_Zeng}, Yang et al. proposed a method for discriminating the complete multiplicities of a parametric univariate polynomial by counting the numbers of distinct real/imaginary roots of the repeated gcds (called multiple factors there) of the given polynomials. In the ${\sf ParametricCompleteMultiplicity}$ {algorithm}, we use the incremental gcds of $\boldsymbol{F} = (P^{(0)},P^{(1)},\ldots,P^{(n)})$ instead of repeated gcds. Their equivalence seems to be quite straightforward. 
However, we could not find a reference for its proof and thereby we provide a detailed one in Appendix B for readers' reference (see Lemma \ref{lem:equiv_rpgcd_icgcd}).
\end{remark}

\section{Comparison}\label{sec:comparison}

In this section, we compare the ``sizes" of the conditions generated by the Algorithm {\sf Parametric\-Complete\-Multiplicity} and those by  \cite{1996_Yang_Hou_Zeng} (abbreviated as YHZ's method).  

\subsection{Review on YHZ's condition}
We start by reproducing the result given by YHZ's method
\label{lem:YHZ} for readers' convenience. 

Assume $P$ is of degree $n$ with the complete multiplicity $\boldsymbol{\mu_c}=(\boldsymbol{\mu_R};\boldsymbol{\mu_I})$. Let $\boldsymbol{\mu}$ be the complex multiplicity of $P$. 
For the sake of simplicity, we introduce the following shorthand notations.

\begin{notation}\ 
        \label{notation:comparison}
\begin{itemize}

\item $\tilde{G}_{i}=\left\{
\begin{array}
[c]{lll}%
P & \text{if } & i=0;\\
R_{\bar{\mu}_{i}}\left(  \tilde{G}_{i-1},\tilde{G}_{i-1}'\right)  & \text{if} & i>0;
\end{array}
\right.  $

\item $s_{i}=\sum\limits_{k=i+1}^{\mu_1}\bar{\mu}_k$, i,e., $s_i$ is the degree of $\tilde{G}_i$ in $x$;

\item Let $\overline{R_{k}\left(
\tilde{G}_i,\tilde{G}'_i\right)}$ represent the principal coefficient of $R_{k}\left(
\tilde{G}_i,\tilde{G}'_i\right)  $, that is, the coefficient of the term $x^{s_i-k}$ in $R_{k}\left(
\tilde{G}_i,\tilde{G}'_i\right)  $.
\end{itemize}
\end{notation}

It is easy to see that 
$\tilde{G}_i$ is the multiple factor of $P$ at the $i$-th level. 
By Theorem \ref{thm:YHZ}, we have%
\begin{align}
&\operatorname*{cmult}\left(  P\right)  =(\boldsymbol{\mu_R};\boldsymbol{\mu_I})\notag\\
\Longleftrightarrow
\ \ \ &\left(  \bigwedge\limits_{i=0}^{\mu_{1}-2}\ \ \bigwedge\limits_{j=\bar{\mu}_{i+1}+1}%
^{s_i}D_{j}(\tilde{G}_i,\tilde{G}_{i}')=0\right)  \ \ \ \wedge
\ \ D_{s_{\mu_1-1}\left(  \tilde{G}_{\mu_{1}-1},\tilde{G}_{\mu_{1}-1}'\right)  }\neq0\notag\\
 \wedge\ &\left(  \bigwedge\limits_{i=0}^{\mu_{1}-1}\ \ {\rm Var}\left(
 D_1(\tilde{G}_i,\tilde{G}_{i}'),\ldots,
 D_{\bar{\mu}_{i+1}}(\tilde{G}_i,\tilde{G}_{i}')
 \right)=\bar{\mu}_{I,i+1}/2
\right)\label{eqs:YHZcond}
\end{align}

Theoretical analysis on YHZ's method in the rest of this section is based on the following natural assumption which is also assumed in  \rm\cite{2021_Hong_Yang_arxiv}.
\begin{assumption}
\label{assumption:nonzero} 
We will assume that  $\forall\,\boldsymbol{\mu}\in\mathcal{M}(n)$, $\overline{R_{s_i}\left(
\tilde{G}_{i},\tilde{G}_{i}'\right)  }$ in the above YHZ's condition is not identically $0$ as a polynomial
in terms of the coefficients of $P$.
\end{assumption}

\subsection{Comparison}
We compare the ``sizes" of the polynomial conditions generated by the ${\sf ParametricCompleteMultiplicity}$  algorithm and those given by YHZ's method from the following two aspects:
\begin{itemize}
\item the numbers of polynomials for partitioning the parameter set, and 
\item the degrees of polynomials appearing in the conditions.  
\end{itemize}
Furthermore, we compare the computational efficiency of the two methods via practical examples.

\subsubsection{Comparison on the number of polynomials}
For simplicity, we introduce the following notations:
\begin{itemize}
\item $T_{\rm YHZ}(n)$:  the number of polynomials appearing in YHZ's condition;
\item $T_{\rm QXY}(n)$:  the number of polynomials appearing in the new condition (produced by Algorithm \ref{alg:rpm}).
\end{itemize}
As usual, we assume $P$ to be a univariate polynomial of degree $n\ge2$ with formal coefficients. 

We start by deducing a recursive formula for computing $T_{\rm YHZ}(n)$.
In YHZ's method, the main idea is to restore the complete multiplicity of $P$ by counting the number of distinct real/imaginary roots of $\tilde{G}_0,\ldots,\tilde{G}_{n-1}$. Thus one only needs to discriminate the cases when $\tilde{G}_i$ has different numbers of distinct real roots/pairs of imaginary roots. For this purpose, only the polynomials in the discriminant sequence of $\tilde{G}_i$'s are needed. 

In the first step, we compute  all the potential gcds of $\tilde{G_0}$ and $\tilde{G}_0'$, i.e.,  $R_n(\tilde{G}_0,\tilde{G}_0')$, \ldots, $R_{1}(\tilde{G}_0,\tilde{G}_0')$. We also need $D_1(\tilde{G}_0,\tilde{G}_0'),\ldots,D_n(\tilde{G}_0,\tilde{G}_0')$ to discriminate the numbers of real roots/pairs of imaginary roots $\tilde{G}_0$ may have.
In this step, we need $n$ polynomials in the parameters.
Note that each $R_i(\tilde{G}_0,\tilde{G}_0')$ is a potential gcd of $\tilde{G}_0$ and $\tilde{G}_0'$ with degree $n-i$ where $i=1,\ldots,n-1$. When $i=n-1$, $\deg R_i(\tilde{G}_0,\tilde{G}_0')=1$. We immediately know that $R_i(\tilde{G}_0,\tilde{G}_0')$ has only one (real) root and there are no other cases to be discriminated. So we do not need extra polynomials. 
 On the other hand, to discriminate all the complete multiplicities of $\tilde{G}_i$, we need $T_{\rm YHZ}(n-i)$ polynomials. As a result, when $n\ge2$, we have
\[T_{\rm YHZ}(n)=n+\sum_{i=1}^{n-2}T_{\rm YHZ}(n-i)=n+\sum_{i=2}^{n-1}T_{\rm YHZ}(i)\]
One can easily derive from the relation that when $n\ge3$,
\[T_{\rm YHZ}(n)-T_{\rm YHZ}(n-1)
=1+T_{\rm YHZ}(n-1)
\]
Therefore,
\[
T_{\rm YHZ}(n)
=\frac{3}{4}\cdot2^n-1
\]

Now we consider the number of polynomials required by the new method to discriminate all the complete multiplicities of $P$, which is denoted by $T_{\rm QXY}(n)$. To discriminant all the complex multiplicities
of $P$, we need $R_{\boldsymbol{\mu}}$'s where $\boldsymbol{\mu}\in\mathcal{M}(n)$ and the total number of  $R_{\boldsymbol{\mu}}$'s is $\#\mathcal{M}(n)$. To discriminate all the complete multiplicities of $P$, we need to consider all the possibilities of $G_i$, which form a set $\{R_{\boldsymbol{\delta}}: \boldsymbol{\delta}\in\mathcal{N}(n)\}$ by Theorem \ref{thm:icgcd} with $$\mathcal{N}(n)=\left\{(\delta_1,\ldots,\delta_t):\,\delta_1+\cdots+\delta_t< n, \delta_1\ge\ldots\ge\delta_t\ge1\right\}.$$ 
Since $\mathcal{N}(n)=\bigcup_{i=0}^{n-1}\mathcal{M}(i)$, we have \[\{R_{\boldsymbol{\delta}}: \boldsymbol{\delta}\in\mathcal{N}(n)\}=\bigcup_{i=0}^{n-1}\{R_{\boldsymbol{\delta}}: \boldsymbol{\delta}\in\mathcal{M}(i)\}.\]
In each subset of the right-hand side, $R_{\boldsymbol{\delta}}$'s have the same degree $n-i$. When $i=n-1$, $R_{\boldsymbol{\delta}}$'s are of degree $1$ and we will not add any polynomial to the condition. When $0\le i\le n-2$, to discriminate all the possible numbers of distinct real/imaginary roots for each $R_{\boldsymbol{\delta}}$, we need $D_1(R_{\boldsymbol{\delta}}'), R_{\boldsymbol{\delta}}, \ldots, D_{n-i}(R_{\boldsymbol{\delta}}, R_{\boldsymbol{\delta}}')$ and thus altogether $\#\mathcal{M}(i)\cdot (n-i)$ polynomials are required to discriminate the numbers of distinct real/imaginary roots for all $R_{\boldsymbol{\delta}}$'s of degree $n-i$. In summary, in the condition generated by the proposed method, the number of polynomials is
\[T_{\rm QXY}(n)=\#\mathcal{M}(n)+\sum_{i=0}^{n-2}\#\mathcal{M}(i)\cdot (n-i)\]

\begin{table}[htb]
\begin{center}
\caption{Comparison on the numbers of polynomials in the conditions generated by YHZ's method and the proposed method, respectively.}\label{tab:comparison_num_polys}
\begin{tabular}{|c|r|r|r|r|r|r|r|r|}\hline
$n$ & 3 & 4 & 5 & 6 & 7 & 8 & 9 & 10  \\\hline
$T_{\rm YHZ}$ & 5& 11& 23& 47& 95& 191& 383& 767 \\\hline
$T_{\rm QXY}$ & 8& 16& 28& 49& 79& 127& 195& 296 \\\hline
$\dfrac{T_{\rm YHZ}}{T_{\rm QXY}}$ & 0.625& 0.688& 0.821& 0.959& 1.203& 1.504& 1.964& 2.591 \\\hline\hline
$n$ & 11 & 12 & 13 & 14 & 15 & 16 & 17 & 18  \\\hline
$T_{\rm YHZ}$ & 1535& 3071& 6143& 12287& 24575& 49151& 98303& 196607 \\\hline
$T_{\rm QXY}$ & 437& 639& 914& 1297& 1812& 2510& 3436& 4670 \\\hline
$\dfrac{T_{\rm YHZ}}{T_{\rm QXY}}$ & 3.513& 4.806& 6.721& 9.473& 13.562& 19.582& 28.610& 42.100 \\\hline
\end{tabular} 
\end{center}
\end{table}

In Table \ref{tab:comparison_num_polys}, we list the two numbers, i.e., $T_{\rm YHZ}$ and $T_{\rm QXY}$, for $n=3,\ldots,18$.
In Figure \ref{fig:comparison_num_polys}, we show the trend of the two numbers as $n$ increases. It is observed that
\begin{itemize}
\item 
As expected, $T_{\rm YHZ}$ increases exponentially with $n$ while $T_{\rm QXY}$ is not. 
In fact, it is seen from Figure \ref{fig:number} that $T_{\rm QXY}$ increases significantly slower than $T_{\rm YHZ}$ as $n$ increases. 

\item From  Fig. \ref{fig:ratio}, we can further see that the ratio of $T_{\rm YHZ}$  to  $T_{\rm QXY}$ is always greater than $1$ for $n\ge7$, which indicates that $T_{\rm QXY}<T_{\rm YHZ}$ when $n\ge7$.

\item Moreover, Fig. \ref{fig:ratio} exhibits an exponential increase on the ratio  $T_{\rm YHZ}/T_{\rm QXY}$ along with $n$. Thus,  $T_{\rm YHZ}$ becomes dramatically smaller than $T_{\rm YHZ}$ when $n$ grows.
\end{itemize}

\begin{figure}[htb]
  \centering
  \begin{subfigure}[b]{0.49\textwidth}
    \includegraphics[width=\textwidth]{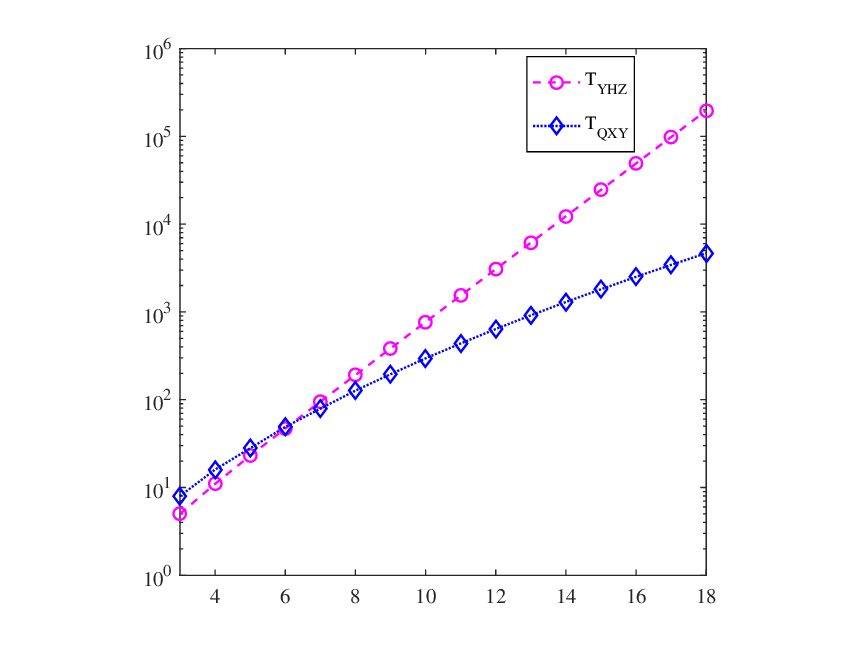}
    \caption{The plots of numbers of polynomials in the conditions generated by the two methods}
    \label{fig:number}
  \end{subfigure}
  \hfill 
  \begin{subfigure}[b]{0.49\textwidth}
    \includegraphics[width=\textwidth]{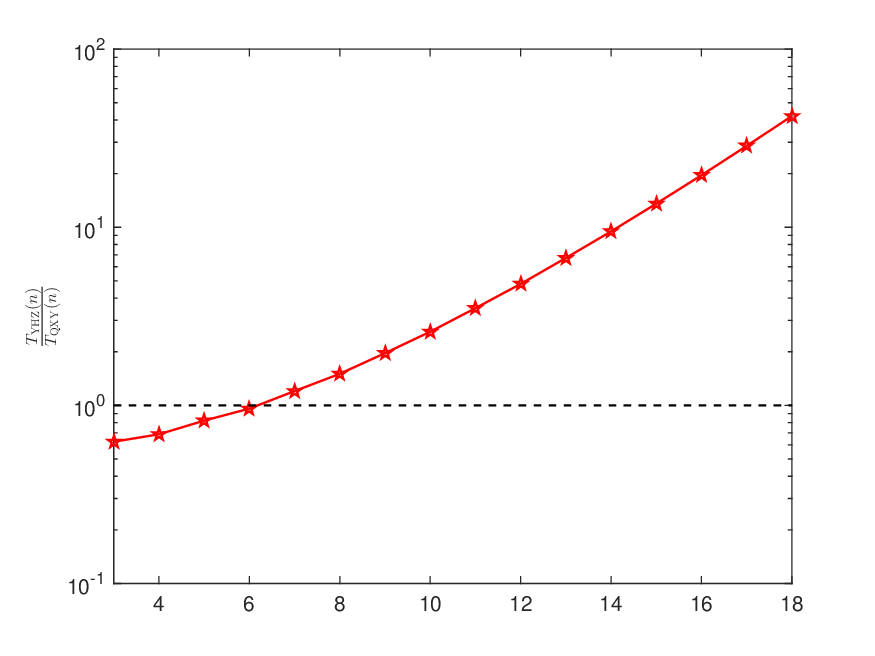}
    \caption{The ratio of the two numbers of polynomials in the conditions generated by the two methods }
    \label{fig:ratio}
  \end{subfigure}
  \caption{Comparison on the numbers of polynomials in the conditions generated by YHZ's method and the proposed method.}\label{fig:comparison_num_polys}
\end{figure}

\subsubsection{Comparison on the degree of polynomials}

In this subsubsection, we compare the degrees of polynomials appearing in the conditions. We choose the maximal degree as an indicator for this assessment. For simplicity, we introduce the following notations:
\begin{itemize}
\item $d_\text{YHZ}$: the maximal degree of polynomials appearing in YHZ's condition;
\item $d_\text{QXY}$: the maximal degree of polynomials appearing in the new condition (produced by Algorithm \ref{alg:rpm}).
\end{itemize}

Under Assumption \ref{assumption:nonzero}, we have the following proposition.

\begin{proposition} \label{prop:maxdegree_YHZ}
 $d_{\rm YHZ}\ge3^{\lfloor{n}/{2}\rfloor}$.  
\end{proposition}

\begin{proof}
Note that when $1\le j\le \deg\tilde{G}_i-1$,
\[D_j\left(  \tilde{G}_{i},\tilde{G}_i'\right)=\pm \operatorname*{lcoef}(\tilde{G}_i)\cdot \overline{R_{j}\left(  \tilde{G}_{i},\tilde{G}_i'\right)  }\]
Thus $\deg_a D_j\left(  \tilde{G}_{i},\tilde{G}_i'\right)\ge\deg_a \overline{R_{j}\left(  \tilde{G}_{i},\tilde{G}_i'\right)  }$. Thereby,
\begin{align*}
d_{\rm YHZ}\ =&\ \max_{\boldsymbol{\mu}\in\mathcal{M}(n)}\max\Bigg(  \bigcup\limits_{i=0}^{\mu_{1}-2}\ \ \bigcup\limits_{j=\bar{\mu}_{i+1}+1}%
^{s_i}\left\{\deg_aD_j\left(  \tilde{G}_{i},\tilde{G}_i'\right)\right\} \ \cup
\ \ \left\{\deg D_{s_{\mu_1-1}}\left(  \tilde{G}_{\mu_{1}-1},\tilde{G}_{\mu_{1}-1}'\right) \right\}\notag\\
 &\ \qquad\qquad\qquad\cup\ \bigcup\limits_{i=0}^{\mu_{1}-1}\ \bigcup_{j=1}^{\bar{\mu}_{i+1}}
 \left\{\deg_a D_j\left(  \tilde{G}_{i},\tilde{G}_i'\right)\right\}
\Bigg)\\
\ge&\ \max_{\boldsymbol{\mu}\in\mathcal{M}(n)}\max\Bigg(  \bigcup\limits_{i=0}^{\mu_{1}-2}\ \ \bigcup\limits_{j=\bar{\mu}_{i+1}+1}%
^{s_i}\left\{\deg_a\overline{R_{j}\left(  \tilde{G}_{i},\tilde{G}_i'\right)  }\right\} \ \cup
\ \ \left\{\deg\overline{R_{s_{\mu_1-1}}\left(  \tilde{G}_{\mu_{1}-1},\tilde{G}_{\mu_{1}-1}'\right)  }\right\}\notag\\
 &\ \qquad\qquad\qquad\cup\ \bigcup\limits_{i=0}^{\mu_{1}-1}\ \bigcup_{j=1}^{\bar{\mu}_{i+1}}
 \left\{\deg_a D_j\left(  \tilde{G}_{i},\tilde{G}_i'\right)\right\}
\Bigg)
\end{align*}
For the sake of simplicity, we use the following shorthand
notation:
\begin{align*}
d_{\boldsymbol{\mu}}\ =&\ \max\Bigg(  \bigcup\limits_{i=0}^{\mu_{1}-2}\ \ \bigcup\limits_{j=\bar{\mu}_{i+1}+1}%
^{s_i}\left\{\deg_a\overline{R_{j}\left(  \tilde{G}_{i},\tilde{G}_i'\right)  }\right\} \ \cup
\ \ \left\{\deg\overline{R_{s_{\mu_1-1}}\left(  \tilde{G}_{\mu_{1}-1},\tilde{G}_{\mu_{1}-1}'\right)  }\right\}\notag\\
 &\ \qquad\quad\cup\ \bigcup\limits_{i=0}^{\mu_{1}-1}\ \bigcup_{j=1}^{\bar{\mu}_{i+1}}
 \left\{\deg_a D_j\left(  \tilde{G}_{i},\tilde{G}_i'\right)\right\}
\Bigg)
\end{align*}
Then $d_{\rm YHZ}\ge\max_{\boldsymbol{\mu}\in\mathcal{M}(n)}d_{\boldsymbol{\mu}}$.

By \cite[Lemma 17-3]{2021_Hong_Yang_arxiv}, when $2\le m\le n-2$,
$d_{\boldsymbol{\mu}}\ge2n+3^{\mu_2}-4\mu_2$.
Thus
\begin{align*}
d_{\rm YHZ}\ &\ge\ \max_{\boldsymbol{\mu}\in\mathcal{M}(n)}d_{\boldsymbol{\mu}}\\
 &\ge\ \max_{\substack{\boldsymbol{\mu}\in\mathcal{M}(n)\\2\le m\le n-2}}d_{\boldsymbol{\mu}}\\
 &\ge\ 2n+3^{\mu_2}-4\mu_2\\ 
 &=\ 2n+3^{\lfloor{n}/{2}\rfloor}-4\lfloor{n}/{2}\rfloor\\
 &\ge\ 3^{\lfloor{n}/{2}\rfloor}
\end{align*}
\end{proof}

For the maximal degree of polynomials in the new condition presented in this paper, we  prove the following proposition.
\begin{proposition}\label{prop:maxdegree_QXY}
When $n\ge3$, $d_{\rm QXY} = n(n-1).$
\end{proposition}

\begin{proof}

Recall \eqref{eqs:QXYcond}. We have
\begin{align*}
d_\text{QXY}\ =&\ \max_{\boldsymbol{\mu}\in\mathcal{M}(n)}\Bigg(\bigcup_{\substack{\boldsymbol{\gamma}\succ\boldsymbol{\bar{\mu}}\\\boldsymbol{\gamma}\in\mathcal{M}(n)}}\left\{
\deg_aR_{\boldsymbol{\gamma}}(\boldsymbol{F})\right\} \ \cup
\ \{\deg_aR_{\boldsymbol{\bar{\mu}}}(\boldsymbol{F})\}\ \cup\ \bigcup\limits_{i=0}^{\mu_{1}-1}\bigcup\limits_{j=1}^{\bar{\mu}_{i+1}}\left\{\deg_a D_j(G_i,G_{i}')\right\}\Bigg)\\
=&\ \max_{\boldsymbol{\mu}\in\mathcal{M}(n)}\Bigg(\max_{\substack{\boldsymbol{\gamma}\succeq\boldsymbol{\bar{\mu}}\\\boldsymbol{\gamma}\in\mathcal{M}(n)}}
\deg_aR_{\boldsymbol{\gamma}}(\boldsymbol{F}), \ \ \ \max_{0\le i\le\mu_{1}-1}\deg_a D_{\bar{\mu}_{i+1}}(G_i,G_{i}')\Bigg)
\end{align*}
It was shown in \cite{2024_Hong_Yang:non-nested} that \[\max_{\substack{\boldsymbol{\gamma}\succeq\boldsymbol{\bar{\mu}}\\\boldsymbol{\gamma}\in\mathcal{M}(n)}}
\deg_aR_{\boldsymbol{\gamma}}(\boldsymbol{F})=2n-\mu_m\]
Thus
\begin{align*}
d_\text{QXY}\ &=\ \max_{\boldsymbol{\mu}\in\mathcal{M}(n)}\bigg(2n-\mu_m,\ \  \max_{0\le i\le\mu_{1}-1}\deg_a D_{\bar{\mu}_{i+1}}(G_i,G_{i}')\bigg)\\
&=\ \max\bigg(2n-1,\ \  \max_{\boldsymbol{\mu}\in\mathcal{M}(n)}\max_{0\le i\le\mu_{1}-1}\deg_a D_{\bar{\mu}_{i+1}}(G_i,G_{i}')\bigg)\\
&=\ \max\bigg(2n-1,\ \  \max_{\boldsymbol{\mu}\in\mathcal{M}(n)}\max_{0\le i\le\mu_{1}-1}2\bar{\mu}_{i+1}\deg_a G_i\bigg)
\end{align*}

Note that 
\[
\deg_a G_i=\left\{
\begin{array}{ll}
1 &  \text{if $i=0;$}\\
(\bar{\mu}_1-1)+(\bar{\mu}_1+\cdots+\bar{\mu}_i) & \text{if $i>0$}. \\
\end{array}
\right.
\]
Hence we can derive that 
\begin{align*}
&\ \max_{\boldsymbol{\mu}\in\mathcal{M}(n)}\max_{0\le i\le\mu_{1}-1}\deg_a D_{\bar{\mu}_{i+1}}(G_i,G_{i}')\\ =&\ 2\max_{\boldsymbol{\mu} \in \mathcal{M}(n)}\max\left(\bar{\mu}_{1},\ \max_{1\le i\le\mu_{1}-1}\bar{\mu}_{i+1}\big((\bar{\mu}_1-1)+(\bar{\mu}_1+\cdots+\bar{\mu}_i)\big)\right)
\end{align*}

Now we consider the following three cases.
\begin{enumerate}
\item 
When $m=1$, i.e., $\boldsymbol{\mu}=(n)$, $\bar{\mu}_1=\cdots=\bar{\mu}_{n}=1$. It follows that
\begin{align*}
&\max\left(\bar{\mu}_{1},\ \max_{1\le i\le\mu_{1}-1}\bar{\mu}_{i+1}\big((\bar{\mu}_1-1)+(\bar{\mu}_1+\cdots+\bar{\mu}_i)\big)\right)=\ \max(1,\ \max_{1\le i\le\mu_{1}-1}i)\ =\ n-1
\end{align*}
Thus,
\[\max_{0\le i\le\mu_{1}-1}\deg_a D_{\bar{\mu}_{i+1}}(G_i,G_{i}')\ =\ 2(n-1)\ < \ 2n-1\]

\item 
When $m=n$, i.e., $\boldsymbol{\mu}=(1,\ldots,1)$, 
\[\max_{0\le i\le\mu_{1}-1}\deg_a D_{\bar{\mu}_{i+1}}(G_i,G_{i}')\ =\ \deg_a D_{\bar{\mu}_{1}}(G_0,G_{0}')\ =\ 2n>2n-1\]

\item When  $2\le m\le n-1$, i.e.,  $\boldsymbol{\mu}\notin\{(n), (1,\ldots,1)\}$, let $b=\mu_1-1$ and $c=\bar{\mu}_{\mu_1}$. Thus $\bar{\mu}_1+\cdots+\bar{\mu}_{\mu_1-1} = n-c$. From the Young tableau of $\boldsymbol{\mu}$ below,
\ytableausetup{centertableaux}
\begin{equation}\label{eq:ytableau}
n-c\left\{
\begin{array}{c}
\\[4em]
\end{array}
\right.
\begin{ytableau}
        \none & \none & \none[b+1] & \none[]& \none[] \\
        \none[\bar{\mu}_1]&\none[] &  &  & &&& \\
        \none[\bar{\mu}_2] &\none[] &  &  & && \none\\
        \none & \none&\none[$\vdots$] & \none[$\vdots$]  \\
        \none[\bar{\mu}_{\mu_1-1}]&\none[] &  &  \\
        \none[\bar{\mu}_{\mu_1}(=c)] &\none[] &  &\none 
\end{ytableau}
\end{equation}
we have $b+c\le n-c$. Now we consider the following function 
$$f(c,b) = c(b+n-c)$$
in $c$ and $b$ whose feasible region (see Fig. \ref{fig:feasible_region}) is bounded by
\[
\left\{
\begin{array}{ll}
1\le c\le n-1, & \text{a natural condition} \\
1< b+1< n, & \text{since $1<\mu_1< n$} \\
b+1 \le n-c, & \text{see \eqref{eq:ytableau}}
\end{array}
\right.\]

\begin{figure}[ht]
\centering
\includegraphics[width=0.6\textwidth]{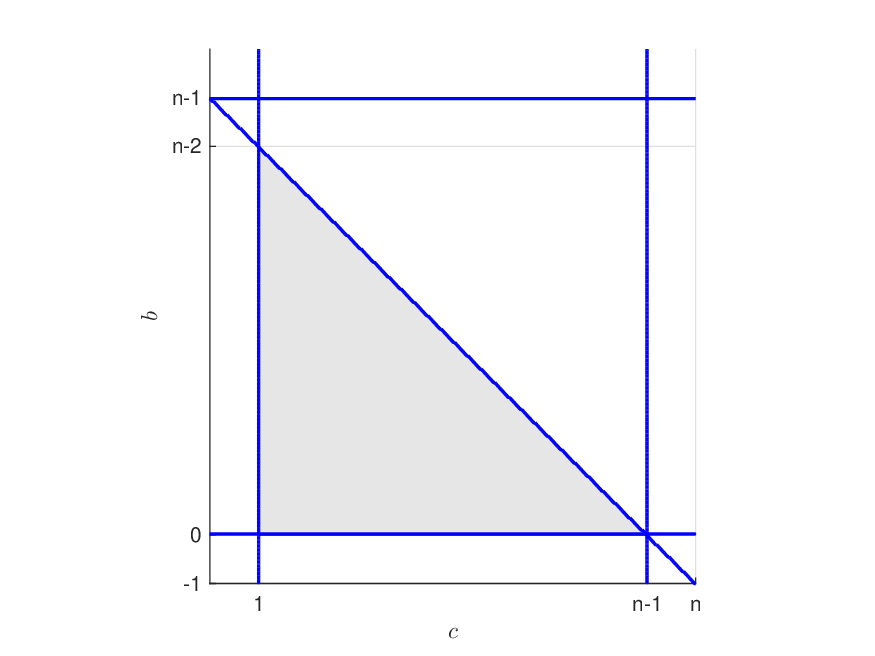}
\caption{The feasible region of $(c,b)$}
\label{fig:feasible_region}
\end{figure}

\noindent Then we have 
\[\dfrac{\partial f}{\partial c}\ =\ (b+n)-2c,\qquad \dfrac{\partial f}{\partial b}\ =\ c.\]
Obviously, $\dfrac{\partial f}{\partial b}\ge1$. Moreover,
\text{since $c\le \bar{\mu}_{\mu_1-1}\le n-c$ which can be deduced from \eqref{eq:ytableau}}, we derive the following:
\[\dfrac{\partial f}{\partial c}=b+\big((n-c)-c\big)\ge0\]
which indicates that $f(c,b)$ increases w.r.t. $c$ and $b$ in the feasible region. Hence, the maximal value of $f$ is achieved on the right boundary of the region, which is defined by $b+1=n-c$. Therefore, the maximal value of $f(c,b)$ over the feasible region is
\begin{align*}
f_{\max}\ =&\ \max_{c\in\{1,\ldots,n\}}f(c,n-c-1)\\
\ =&\ \max_{c\in\{1,\ldots,n\}}c(2n-2c-1)\\
\ =&\ \max_{c\in\{1,\ldots,n\}}\left(-2\Big(c-\frac{2n-1}{4}\Big)^2+\frac{(2n-1)^2}{8}\right).
\end{align*}
We consider the following two cases.
\begin{enumerate}
\item 
When $n$ is even, i.e., $n=2k$ for some integer $k\ge1$, 
\[\max_{c\in\{1,\ldots,n\}}\left(-2\Big(c-\frac{2n-1}{4}\Big)^2+\frac{(2n-1)^2}{8}\right)=\max_{c\in\{1,\ldots,n\}}\left(-2\Big(c-\frac{4k-1}{4}\Big)^2+\frac{(4k-1)^2}{8}\right).\]
The above maximal value is achieved at $c=k=\dfrac{n}{2}$, which immediately yields
\[f_{\max}=c(2n-2c-1)\Big|_{c=\frac{n}{2}}=\frac{n(n-1)}{2}.\]
In this case, $\boldsymbol{\mu}=\left(\dfrac{n}{2},\dfrac{n}{2}\right)$.

\item When $n$ is odd, i.e., $n=2k+1$ for some integer $k\ge0$, 
\[\max_{c\in\{1,\ldots,n\}}\left(-2\Big(c-\frac{2n-1}{4}\Big)^2+\frac{(2n-1)^2}{8}\right)=\max_{c\in\{1,\ldots,n\}}\left(-2\Big(c-\frac{4k+1}{4}\Big)^2+\frac{(4k+1)^2}{8}\right).\]
The above maximal value is achieved at $c=k=\dfrac{n-1}{2}$, which immediately yields
\[f_{\max}=c(2n-2c-1)\Big|_{c=\frac{n-1}{2}}=\frac{n(n-1)}{2}.\]
In this case, $\boldsymbol{\mu}=\left(\dfrac{n+1}{2},\dfrac{n-1}{2}\right)$.
\end{enumerate}
\end{enumerate}

Therefore,
\begin{align*}
d_\text{QXY}\ &=\ \max\bigg(2n,\ \  2\max_{\substack{\boldsymbol{\mu} \in \mathcal{M}(n)\\\boldsymbol{\mu}\notin\{(n),(1,\ldots,1)\}}} \bar{\mu}_{\mu_1}(\bar{\mu}_1-1+\bar{\mu}_1+\cdots+\bar{\mu}_{\mu_1-1})\bigg)\\
&=\ \max\big(2n,\ \  n(n-1) \big).
\end{align*}

Since $2n\le n(n-1)$ for $n\ge3$, we conclude that $d_{\rm QXY}\ =\ n(n-1)$.
\end{proof}

From Propositions \ref{prop:maxdegree_YHZ} and \ref{prop:maxdegree_QXY}, it is immediately seen that the maximal degree of polynomials in the condition given by YHZ's method increases in exponential scale w.r.t. $n$ while that given by QXY's method shows an increase in polynomial scale.
Furthermore, with the help of Propositions \ref{prop:maxdegree_YHZ} and \ref{prop:maxdegree_QXY}, one can easily show the following proposition.

\begin{proposition}
When $n\ge8$, $d_{\rm YHZ}\ge d_{\rm QXY}$.
\end{proposition}

In Table \ref{tab:comparison_maxdeg}, we list the two numbers, i.e., $d_{\rm YHZ}$ and $d_{\rm QXY}$, for $n=3,\ldots,8,$ while
Fig. \ref{fig:comparison_maxdeg} shows the trend of the two numbers as $n$ increases. From Table \ref{tab:comparison_maxdeg} and
Fig. \ref{fig:comparison_maxdeg}, we make the following observations.

\begin{itemize}
\item 
When $d< 6$, $d_{\rm YHZ}\le d_{\rm QXY}$ and their differences is relatively small (no more than 2, see Table \ref{tab:comparison_maxdeg}). 
Roughly speaking, $d_{\rm QXY}\approx d_{\rm YHZ}$.

\item  When $n$ surpasses $6$, $d_{\rm QXY}$ is significantly smaller than $d_{\rm YHZ}$ (see Fig. \ref{fig:maxdeg}).  
 When $n\ge 4$, the ratio $\dfrac{d_{\rm YHZ}}{d_{\rm QXY}}$ grows as $n$ increases (see Fig. \ref{fig:maxdeg_ratio}). Therefore, the size of the polynomial with the maximal degree in YHZ's method becomes even bigger than that in the proposed method. Hence, it is expected that when $n$ is large enough, the new method will be  far more efficient than YHZ's method.

\end{itemize}

\begin{table}[htb]
\begin{center}
\caption{Comparison on the maximal degrees of polynomials in the conditions generated by YHZ's method and the proposed method.}\label{tab:comparison_maxdeg}
\begin{tabular}{|c|r|r|r|r|r|r|}\hline
$n$ & 3 & 4 & 5 & 6 & 7 & 8   \\\hline
$d_{\rm YHZ}$ & 6& 10& 18& 30& 54& 90 \\\hline
$d_{\rm QXY}$ & 6& 12& 20& 30& 42& 56 \\\hline
$\dfrac{d_{\rm YHZ}}{d_{\rm QXY}}$ & 1.000& 0.833& 0.900& 1.000& 1.286& 1.607 \\\hline
\end{tabular} 
\end{center}
\end{table} 

\begin{figure}[htb]
  \centering
  \begin{subfigure}[b]{0.49\textwidth}
    \includegraphics[width=\textwidth]{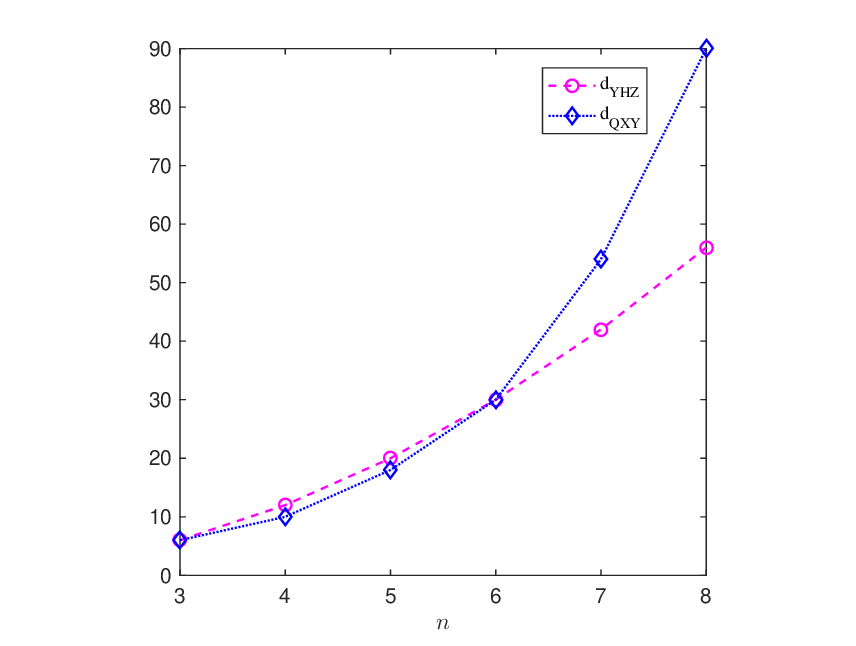}
    \caption{The plots of numbers of polynomials in the conditions generated by the two methods}
    \label{fig:maxdeg}
  \end{subfigure}
  \hfill 
  \begin{subfigure}[b]{0.49\textwidth}
    \includegraphics[width=\textwidth]{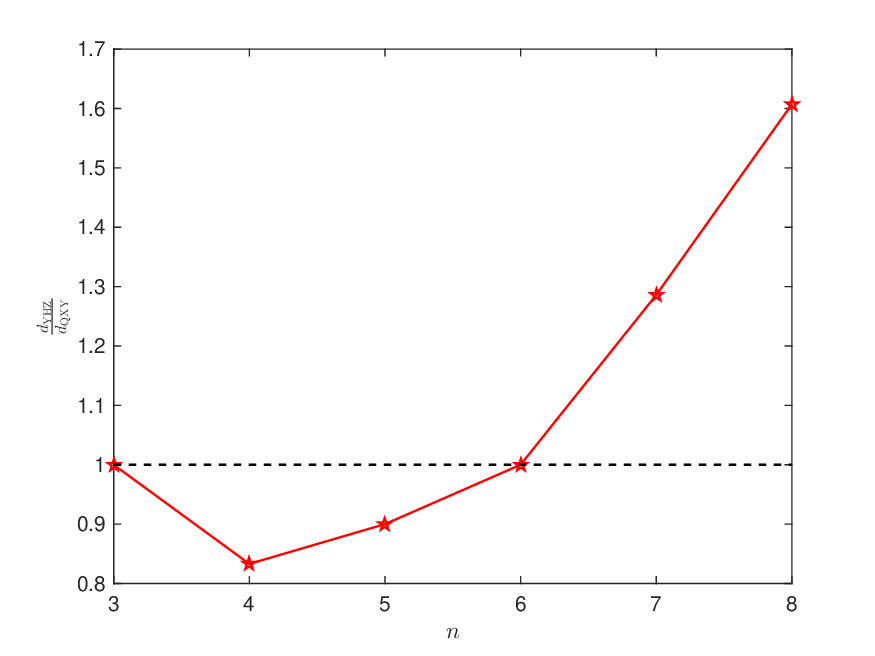}
    \caption{The ratio of the two numbers of polynomials in the conditions generated by the two methods }
    \label{fig:maxdeg_ratio}
  \end{subfigure}
  \caption{Comparison on the maximal degrees of polynomials in the conditions generated by YHZ's method and the proposed method, respectively.}\label{fig:comparison_maxdeg}
\end{figure}

\subsubsection{Comparison on performance}
In this part, we compare YHZ's method to the proposed method (i.e., QXY's method) in terms of computational efficiency.
For this purpose, we carry out several experiments for $n$ ranging from $3$ to $9$. The experiments are performed on a PC equipped with an Intel(R) Core(TM) i5-7300U processor and an 8GB RAM. 

\begin{table}[htb]
\begin{center}
\caption{Comparison on the time cost of YHZ's method and the proposed method (in seconds)}\label{tab:comparison_time}
\begin{tabular}{|c|r|r|r|r|r|r|r|}\hline
$n$ & 3~~ & 4~~ & 5~~ & 6~~ & 7\,~~ & 8~~~~~ & 9  \\\hline
$t_{\rm YHZ}$ & 0.000& 0.015& 0.031& 0.078& 4.046& 8772.687&?\\\hline
$t_{\rm QXY}$ & 0.015& 0.031& 0.093& 0.171& 1.046& 434.468&?\\\hline
$\dfrac{t_{\rm YHZ}}{t_{\rm QXY}}$ & 0.000& 0.484& 0.333& 0.456& 3.868& 20.192&$-$ \\\hline
\end{tabular} 
\end{center}
\end{table}

The experimental results are reported in Table \ref{tab:comparison_time} where
$t_\text{YHZ}$ and $t_\text{QXY}$ represent the time cost (in seconds) charged by YHZ's method and the proposed method, respectively. It is seen that the proposed method outperforms YHZ's method when solving relatively larger examples (see $n=7,8$ in Table \ref{tab:comparison_time}).

Although both methods fail on the case $n=9$, the new algorithm can still give some positive results. For example, if we take the polynomial of the following form 
$$P=x^9+a_7x^7+\cdots+a_0$$
(which can also be viewed as a generic form of polynomials with degree $9$) as input,
 the algorithm {\sf ParametricCompleteMultiplicity} succeeds after 3009.296 seconds ($<1$ hour) while YHZ's method does not terminate after running for 10 hours.


\section{Conclusion}\label{sec:conclusion}
In this paper, we propose a new algorithm for solving the parametric complete multiplicity problem. Different from using the repeated gcd computation in the classical one \cite{1996_Yang_Hou_Zeng},
the new approach uses incremental gcds instead, which can effectively reduce the size of the polynomials appearing in the condition. 

It is noted that in the new approach, we also confront with nested determinants/subresul\-tants,
which requires intensive computation. Thus a natural question is: is it possible to use non-nested determinants to solve the parametric complete multiplicity problem? The investigation along this direction is very attractive not only because it can help to improve the efficiency of the algorithms but also provides a closed form for the parametric complete multiplicity problem.
Another interesting direction is to explore the hidden structures for subresultant used in the algorithm ${\sf ParametricCompleteMultiplicity}$ which may further enhance the computational efficiency.

\bigskip
\noindent{\bf Acknowledgements.} Bican Xia's work was supported by the National Key R \& D Program of China (No. 2022YFA1005102). Simin Qin and Jing Yang's work was supported by
National Natural Science Foundation of China (Grant Nos.: 12326353 and 12261010) and the National Science Cultivation Project of GXMZU (Grant No.: 2022MDKJ001).

\section*{Appendix A.}

Appendix A is dedicated to proving Proposition \ref{prop:sres_prem}. For this purpose, we need the following lemma which is a specialization of \cite[Lemma 29]{2024_Hong_Yang}. 

\begin{lemma}\label{lem:R_lincomb}
Given $P\in\mathbb{R}[x]$ with degree $n$, let $\boldsymbol{F} = (P^{(0)},P^{(1)},\ldots,P^{(t)})$ and $\boldsymbol{\delta} = (\delta_{1}, \ldots, \delta_{t})\in\mathcal{P}(n,t)$. Then we have
$R_{\boldsymbol{\delta}}\in\langle \boldsymbol{F}\rangle$. More explicitly, 
$$R_{\boldsymbol{\delta}} = \sum_{u = 0}^{t}\sum_{v = 0}^{\delta_{u} - 1}c_{u,v}x^{v}P^{(u)}$$
where $\delta_0$ is determined as in \eqref{eqs:delta0}, $c_{u,\delta_{u} - 1}=(-1)^{\sigma}a_n^{\tau}\overline{R_{\boldsymbol{\delta}-\boldsymbol{e}_{u}}}$ for $u\ge1$  and 
\[\sigma=1+\delta_u+\cdots+\delta_t,\ \ \text{and}\ \ \ 
\tau=\left\{
\begin{array}{ll}
1 &\text{if}\ d_u+\delta_u>\max\limits_{\substack{\delta_i\ne0\\i\ne 0,u}}(d_i+\delta_i);\\
0&\text{otherwise}.
\end{array}\right.\]
\end{lemma}

\begin{proof}[Proof of Proposition \ref{prop:sres_prem}]
Recall $\boldsymbol{\delta}=(\delta_1,\ldots,\delta_t)$ satisfies $|\boldsymbol{\delta}|\le n$, $\delta_t\ge 1$ and $\delta_j>\delta_{j+1}$ for some $1\le j<t$. 
Let $\boldsymbol{\gamma}=\boldsymbol{\delta}-\boldsymbol{e}_j$ and $\boldsymbol{\gamma'}=\boldsymbol{\delta}-\boldsymbol{e}_t$. 
By Lemma \ref{lem:R_lincomb},
\begin{align*}
R_{\boldsymbol{\gamma}}=\sum_{u=0}^{t}\sum_{v = 0}^{\gamma_{u} - 1}c_{u,v}x^{v}P^{(u)},\quad R_{\boldsymbol{\gamma'}}=\sum_{u=0}^{t}\sum_{v = 0}^{\gamma_{u}' - 1}c'_{u,v}x^{v}P^{(u)}
\end{align*}
where $\gamma_0$ and $\gamma_0'$ are determined as in \eqref{eqs:delta0},
and 
\begin{align}
c_{t,\gamma_t-1}&=(-1)^{1+\gamma_t}\overline{R_{\boldsymbol{\gamma}-\boldsymbol{e}_t}}\notag\\
&\quad\text{($\tau=0$ is implied by $\gamma_t<\gamma_1$ and $d_t=n-t<n-1=d_1$)}\notag\\
&=(-1)^{1+\delta_t}\overline{R_{\boldsymbol{\delta}-\boldsymbol{e}_j-\boldsymbol{e}_t}(\boldsymbol{F})},\hspace{9em}\text{(since $\gamma_t=\delta_t$)}\label{eq:lc_t}\\
c'_{j,\gamma_j'-1}&=(-1)^{1+\gamma'_j+\cdots+\gamma'_t}a_n^{\tau}\overline{R_{\boldsymbol{\gamma'}-\boldsymbol{e}_j}}\notag\\
&=(-1)^{\delta_j+\cdots+\delta_t}a_n^{\tau}\overline{R_{\boldsymbol{\delta}-\boldsymbol{e}_t-\boldsymbol{e}_j}}
\label{eq:lc_j}\\
&\hspace{1em}\ \text{(since $\gamma_j'=\delta_j$ for $j<t$ and $\gamma_t'=\delta_t-1$)},\notag
\end{align}
where $\tau=\left\{
\begin{array}{ll}
1 &\text{if}\ j=1;\\
0&\text{otherwise}
\end{array}\right.
$ (which is a specialization of $\tau$ in Lemma \ref{lem:R_lincomb}).

Since $\gamma_t=\delta_t$ and $\gamma_j'=\delta_j$, the above linear combinations can be further rewritten as
\begin{align}
&c_{t,\delta_t-1}x^{\delta_t-1}P^{(t)}=R_{\boldsymbol{\gamma}}-\sum_{(u,v)\in C_1}c_{u,v}x^{v}P^{(u)},\label{eqs:C_t}\\ &c'_{j,\delta_j-1}x^{\delta_j-1}P^{(j)}=R_{\boldsymbol{\gamma'}}-\sum_{(u,v)\in C_2}c'_{u,v}x^{v}P^{(u)}.\label{eqs:C_j}
\end{align}
where 
\begin{align*}
C_1&=\{(u,v):\,0\le u\le t, 0\le v\le \gamma_{u}-1\}\backslash\{(t,\delta_t-1)\},\\
C_2&=\{(u,v):\,0\le u\le t, 0\le v\le \gamma'_{u}-1\}\backslash\{(j,\delta_j-1)\}.\\
\end{align*}
Now we consider $R_{\boldsymbol{\delta}}$, which is
\begin{align*}
R_{\boldsymbol{\delta}}=\operatorname{dp}(&x^{\delta_1-2}P^{(0)},\ldots,x^0P^{(0)},x^{\delta_1-1}P^{(1)},\ldots,x^0P^{(1)},\\
&\ldots,x^{\delta_j-1}P^{(j)},\ldots,x^0P^{(j)},\ldots\ldots,x^{\delta_t-1}P^{(t)},\ldots,x^0P^{(t)}).
\end{align*}
For the sake of simplicity, we introduce the following notations. Let
$$P_{u,v} = x^vP^{(u)}, Q_{u,v} = P_{u,v-1},\dots,P_{u,0}.$$ 
Then we have 
\begin{align*}
R_{\boldsymbol{\delta}} &= \operatorname*{dp}(Q_{0,\delta_1-1},Q_{1,\delta_1},\dots,Q_{j,\delta_j},\dots,Q_{t,\delta_t},\dots,Q_{n,\delta_n}) \\
&=\operatorname*{dp}(Q_{0,\delta_1-1},Q_{1,\delta_1},\dots,\underbrace{P_{j,\delta_j-1},Q_{j,\delta_j-1}}_{Q_{j,\delta_j}},\dots,\underbrace{P_{t,\delta_t-1},Q_{t,\delta_t-1}}_{Q_{t,\delta_t}},\dots,Q_{n,\delta_n}).
\end{align*}
Consider $c_{t,\delta_t-1}c'_{j,\delta_j-1}R_{\boldsymbol{\delta}}$. We derive the following:
\begin{align*}
&c_{t,\delta_t-1}c'_{j,\delta_j-1}R_{\boldsymbol{\delta}} \\
=\ &\operatorname*{dp}(Q_{0,\delta_1-1},Q_{1,\delta_1},\dots,c'_{j,\delta_j-1}P_{j,\delta_j-1},Q_{j,\delta_j-1},\dots,c_{t,\delta_t-1}P_{t,\delta_t-1},Q_{t,\delta_t-1},\dots,Q_{n,\delta_n})\\
&\hspace{25em}(\text{pushing $c_{t,\delta_t-1}c'_{j,\delta_j-1}$ into $\operatorname*{dp}$})\\
=\ &\operatorname*{dp} (Q_{0,\delta_1-1},Q_{1,\delta_1},\dots,R_{\boldsymbol{\delta}-\boldsymbol{e}_t},Q_{j,\delta_j-1},\dots,R_{\boldsymbol{\delta}-\boldsymbol{e}_j},Q_{t,\delta_t-1},\dots,Q_{n,\delta_n}) \\
&\quad(\text{using $\eqref{eqs:C_t}$--$\eqref{eqs:C_j}$ and the multi-linearity of determinant polynomial to simplify $R_{\boldsymbol{\delta}}$})\\
=\ &(-1)^{\delta_j+\dots+\delta_{t-1}} \operatorname*{dp}(Q_{0,\delta_1-1},Q_{1,\delta_1},\dots,Q_{j,\delta_j-1},\dots,Q_{t,\delta_t-1},\dots,Q_{n,\delta_n},R_{\boldsymbol{\delta}-\boldsymbol{e}_j},R_{\boldsymbol{\delta}-\boldsymbol{e}_t}) \\
&\hspace{22em}(\text{moving $R_{\boldsymbol{\delta}-\boldsymbol{e}_j}$ and $R_{\boldsymbol{\delta}-\boldsymbol{e}_t}$ to the end})\\
=\ &(-1)^{\delta_j+\dots+\delta_{t-1}}a_n^{\tau}\overline{R_{\boldsymbol{\delta}- \boldsymbol{e}_j-\boldsymbol{e}_t}} \operatorname*{dp}(R_{\boldsymbol{\delta}-\boldsymbol{e}_j},R_{\boldsymbol{\delta}-\boldsymbol{e}_t})\
\hspace{3.5em}(\text{using the block structure of dp})
\end{align*}
where $$\tau=\left\{
\begin{array}{ll}
1 &\text{if}\ j=1;\\
0&\text{otherwise}.
\end{array}\right.
$$

Now we substitute \eqref{eq:lc_t} and \eqref{eq:lc_j} into the first line of the equation and obtain
\begin{equation}\label{eqs:2bcancelled}
(-1)^{\delta_j+\dots+\delta_{t-1}+1}a_n^{\tau}\left(\overline{R_{\boldsymbol{\delta}- \boldsymbol{e}_j-\boldsymbol{e}_t}(\boldsymbol{F})}\right)^{2} R_{\boldsymbol{\delta}}(\boldsymbol{F})=(-1)^{\delta_j+\dots+\delta_{t-1}}a_n^{\tau}\overline{R_{\boldsymbol{\delta}- \boldsymbol{e}_j-\boldsymbol{e}_t}(\boldsymbol{F})} \operatorname*{dp}(R_{\boldsymbol{\delta}-\boldsymbol{e}_j},R_{\boldsymbol{\delta}-\boldsymbol{e}_t}) \end{equation}
Note that $\boldsymbol{\delta}$ is a decreasing vector and so is $\boldsymbol{\delta}- \boldsymbol{e}_t-\boldsymbol{e}_j$, which implies that for some value of $a_{i}$'s, $\overline{R_{\boldsymbol{\delta}- \boldsymbol{e}_t-\boldsymbol{e}_j}(\boldsymbol{F})}\ne0$. Hence $\overline{R_{\boldsymbol{\delta}- \boldsymbol{e}_t-\boldsymbol{e}_j}(\boldsymbol{F})}$ is a nonzero polynomial in terms of $a_i$'s. Therefore, after canceling the common factors from both sides of \eqref{eqs:2bcancelled}, we have 
\[-\overline{R_{\boldsymbol{\delta}- \boldsymbol{e}_t-\boldsymbol{e}_j}(\boldsymbol{F})}R_{\boldsymbol{\delta}}(\boldsymbol{F})=\operatorname{dp}(R_{\boldsymbol{\gamma}}(\boldsymbol{F}),R_{\boldsymbol{\gamma'}}(\boldsymbol{F})).\]
Note that 
\begin{align*}
\deg R_{\boldsymbol{\gamma}}(\boldsymbol{F})&=n-|\boldsymbol{\gamma}|=n-|\boldsymbol{\delta}|+1,\\
\deg R_{\boldsymbol{\gamma'}}(\boldsymbol{F})&=n-|\boldsymbol{\gamma'}|=n-|\boldsymbol{\delta}|+1.
\end{align*}
which implies $\deg R_{\boldsymbol{\gamma}}(\boldsymbol{F})=\deg R_{\boldsymbol{\gamma'}}(\boldsymbol{F})$. Therefore,  
we deduce that
$$\operatorname{dp}(R_{\boldsymbol{\gamma}}(\boldsymbol{F}),R_{\boldsymbol{\gamma'}}(\boldsymbol{F}))=-\operatorname*{prem}(R_{\boldsymbol{\gamma}}(\boldsymbol{F}),R_{\boldsymbol{\gamma'}}(\boldsymbol{F})) $$
which indicates that
\[\overline{R_{\boldsymbol{\delta}- \boldsymbol{e}_j-\boldsymbol{e}_t}(\boldsymbol{F})}R_{\boldsymbol{\delta}}(\boldsymbol{F})=\operatorname*{prem}(R_{\boldsymbol{\delta}-\boldsymbol{e}_j}(\boldsymbol{F}),R_{\boldsymbol{\delta}-\boldsymbol{e}_t}(\boldsymbol{F})).\]
\end{proof}

\section*{Appendix B.}
Given a   polynomial $P$, the repeated gcd of $P$ and $P'$ is defined as $(\tilde{G}_1,\ldots,\tilde{G}_n)$ where 
\[\tilde{G}_{i} =\gcd(\tilde{G}_{i-1},\tilde{G}'_{i-1})\]
with
$\tilde{G}_{0}:=P$.
In this paper, we use the incremental gcds of $P^{(0)},P^{(1)},\ldots$ instead of repeated gcds. In the appendix, we provide a detailed proof on their equivalence.

\begin{lemma}\label{lem:equiv_rpgcd_icgcd}
 Let $G_0=P$. Then for $i\ge0$, we have $G_{i}\sim \tilde{G}_{i} $.
\end{lemma}

\begin{proof}
 Without loss of generality, we assume 
 $$ P = a_n\prod_{k = 1}^{m} (x - r_{k})^{\mu_{k}}. $$
 To prove
$G_{i} \sim \tilde{G}_{i}$,
 we only need to show
 \begin{itemize}
 \item $G_{i} \sim \prod_{\mu_{k} > i}^{ } (x - r_{k})^{\mu_{k} - i}$, and
 \item $\tilde{G}_{i} \sim \prod_{\mu_{k} > i}^{ } (x - r_{k})^{\mu_{k} - i},$
 \end{itemize}
 whose proof will be given in an inductive manner.
 
 {\bf Base case.}
 When $i = 0$, we have
 \[
 G_{0} =\tilde{G}_{0} = P =a_{n}  \prod_{\mu_{k} > 0}^{ } (x - r_{k})^{\mu_{k}}
 \]
 which is obviously true.

{\bf Induction step.} 
 We assume the claim holds when $i = j$, i.e.
 \begin{align*}
 G_{j} &= \gcd(G_{j - 1} , P^{(j)})\sim \prod_{\mu_{k} > j}^{ } (x - r_{k})^{\mu_{k} - j},  \\
 \tilde{G}_{j} &= \gcd(\tilde{G}_{j - 1}, \tilde{G}'_{j - 1})\sim \prod_{\mu_{k} > j}^{ } (x - r_{k})^{\mu_{k} - j} .
 \end{align*}
 Then for $i = j + 1$, 
 \[
G_{j + 1} = \gcd(G_{j}, P^{(j + 1)}),\qquad 
\tilde{G}_{j + 1} = \gcd( \tilde{G}_{j},  \tilde{G}'_{j}) .
 \]
 Next we deduce the expressions for $G_{j+1}$ and $\tilde{G}_{j + 1}$, respectively.

By assumption, $G_{j} \sim \prod_{\mu_{k} > j}^{ } (x - r_{k})^{\mu_{k} - j}$.
Obviously,  we have
\begin{itemize}
\item $ \prod_{\mu_{k} > j + 1}^{}(x - r_{k})^{\mu_{k} - (j + 1)} \mid  G_{j} $, and
\item when $\mu_{k} > j$,
 $$ (x - r_{k})^{\mu_{k}}\mid P^{(0)}, (x - r_{k})^{\mu_{k} - 1}\mid P^{(1)}, \ldots, (x - r_{k})^{\mu_{k} - j}\mid P^{(j)} $$
\end{itemize}
Next we show that $ \prod_{\mu_{k} > j + 1}^{}(x - r_{k})^{\mu_{k} - (j + 1)} \mid  P^{(j + 1)} $. 
The key for verifying the claim is the observation that when $\mu_{k} > j + 1$, $$P^{(j + 1)}(r_{k}) = \cdots = P^{(\mu_{k} - 1)}(r_{k}) = 0,$$
but $P^{(\mu_{k})}(r_{k}) \neq 0$, which indicates that
 $ x = r_{k} $ is the multiple roots of $P^{(j + 1)}(x)$ with multiplicity  $\mu_{k} - (j + 1) $.  Thereby, $(x - r_{k})^{\mu_{k} - (j + 1)}\mid P^{(j + 1)}(x)$. Moreover, we also have $(x - r_{k})^{\mu_{k} - j}\nmid P^{(j + 1)}$. Hence

 $$ G_{j + 1} = \gcd(G_{j} , P^{(j + 1)}) \sim \prod_{\mu_{k} > j + 1}^{}(x - r_{k})^{\mu_{k} - (j + 1)} $$

Again, by assumption, $\tilde{G}_{j} \sim \prod_{\mu_{k} > j}^{ } (x - r_{k})^{\mu_{k} - j}$.
Thus
 \begin{align*}
\tilde{G}'_{j}\sim \left( \prod_{\mu_{k} > j}^{ } (x - r_{k})^{\mu_{k} - j}\right)' &= \sum_{\mu_{k} > j}\left((\mu_{k} - j)(x - r_{k})^{\mu_{k} - j - 1} \prod_{\underset{\mu_{\ell} > j}{\ell \neq k} }^{}(x - r_{\ell})^{\mu_{\ell} - j}\right) \\
 &= \left(\sum_{\mu_k > j}^{}(\mu_k - j)\prod_{\underset{\mu_{\ell} > j}{\ell \neq k}}(x - r_{\ell})\right) \left(\prod_{\mu_{k} > j}^{}(x - r_{k})^{\mu_{k} - j - 1} \right).
 \end{align*}
 It follows that
 \begin{align*}
 \tilde{G}_{j + 1} &= \gcd( \tilde{G}_{j},  \tilde{G}'_{j}) \\
 &\sim\gcd \left( \prod_{\mu_{k} > j}^{ } (x - r_{k})^{\mu_{k} - j}, \left( \prod_{\mu_{k} > j}^{ } (x - r_{k})^{\mu_{k} - j}\right)'\right) \\
\\
&=\gcd \left( \prod_{\mu_{k} > j}^{ } (x - r_{k})^{\mu_{k} - j}, \left(\sum_{\mu_k > j}^{}(\mu_k - j)\prod_{\underset{\mu_{\ell} > j}{\ell \neq k}}(x - r_{\ell})\right) \left(\prod_{\mu_{k} > j}^{}(x - r_{k})^{\mu_{k} - j - 1} \right)\right)\\
&=\gcd \left( \prod_{\mu_{k} > j}^{ } (x - r_{k})^{}, \sum_{\mu_k > j}(\mu_k - j)\prod_{\underset{\mu_{\ell} > j}{\ell \neq k}}(x - r_{\ell})\right)\cdot \left(\prod_{\mu_{k} > j}^{}(x - r_{k})^{\mu_{k} - j - 1} \right).
\end{align*}
It remains to show that
\[\gcd \left( \prod_{\mu_{k} > j}(x - r_{k}), \sum_{\mu_k > j}(\mu_k - j)\prod_{\underset{\mu_{\ell} > j}{\ell \neq k}}(x - r_{\ell})\right)=1\]
which can be proved with the following verification
$$\left(\sum_{\mu_{k} > j}^{}(\mu_{k} - j)\prod_{\underset{\mu_{l} > j}{l \neq k}}(x - r_{\ell})\right)\Bigg|_{x=r_u} =(\mu_{u} - j)\prod_{\underset{\mu_{\ell} > j}{\ell \neq k}}(r_{u} - r_{\ell}) \neq 0$$
for $u\in\{\mu_k:\,\mu_k>j\}$. Hence 
\[\tilde{G}_{j + 1} \sim\prod_{\mu_{k} > j}^{}(x - r_{k})^{\mu_{k} - j - 1}\sim G_{j+1}. \]
\end{proof}

%

\end{document}